\documentclass[12pt]{article}
\usepackage{amsfonts}
\usepackage{amsmath}
\usepackage{amssymb}
\usepackage{graphicx}
\usepackage{color}
\usepackage[all, knot]{xy}
\usepackage{tikz}
\usepackage{array}
\usepackage{hyperref}
\usepackage{ulem}

\usepackage{comment}

\usepackage[utf8]{inputenc}
\usepackage{epstopdf}
\usepackage{subcaption}
\usepackage{amsthm}
\usepackage{enumitem}
\usepackage{mathrsfs}
\usepackage{mathtools}

\usepackage[margin=3cm]{geometry}

\def \be {\begin{equation}}
\def \ee {\end{equation}}
\def \bea {\begin{eqnarray}}
\def \eea {\end{eqnarray}}
\def \nn {\nonumber}

\def \rr {\raise.35ex\hbox{\small $\prime$}\kern-.17em{\mbox{\large $\imath$}}}

\def \dels {\partial\kern-.6em /\kern.1em}
\def \As {{A\kern-.5em / \kern.5em}}
\def \Ds {D\kern-.7em / \kern.5em}

\def \ks {k\kern-.5em /}
\def \ls {l\kern-.5em /}

\newcommand{\ci}[1]{}
\newcommand{\ba}{\begin{eqnarray}}
\newcommand{\ea}{\end{eqnarray}}
\newcommand{\bal}{\begin{align}}
\newcommand{\eal}{\end{align}}
\newcommand{\bay}[1]{\left(\begin{array}{#1}}
\newcommand{\eay}{\end{array}\right)}

\newcommand{\hide}[1]{}

\newcommand{\fsl}[1]{\ensuremath{\mathrlap{\!\not{\phantom{#1}}}#1}}

\newtheorem{proposition}{Proposition}

\newlist{axioms}{enumerate}{2}
\setlist[axioms,1]{label=\textbf{A\arabic{axiomsi}.}, ref=A\arabic{axiomsi}}
\setlist[axioms,2]{label=\textbf{A\arabic{axiomsi}\rlap{\myEnumCounter{axiomsii}}.},%
                   ref=A\arabic{axiomsi}\myEnumCounter{axiomsii},%
                   align=parleft,%
                   leftmargin=0em,%
                   itemsep=1.4ex,%
                   before={\stepcounter{axiomsi}}}

  \usetikzlibrary{decorations.markings}

\begin{document}
\begin{titlepage}
%\begin{flushright}
%NORDITA 2019-102 %\\
%October,  2019
%\end{flushright}

\begin{center}

\textbf{\LARGE
Chemical Potential and\\ 
Analytic Continuation for\\ 
Non-Hermitian Lattice Fermions
\vskip.3cm
}
\vskip .5in
{\large
Chen-Te Ma$^{a}$ \footnote{e-mail address: yefgst@gmail.com} 
and Hui Zhang$^{b,c}$ \footnote{e-mail address: Mr.zhanghui@m.scnu.edu.cn}
\\
\vskip 1mm
}
{\sl
$^a$ 
Centre for Cosmology and Science Popularization, \\
Shree Guru Gobind Singh Tricentenary University, Gurugram, Haryana 122505, India.
\\
$^b$
State Key Laboratory of Nuclear Physics and Technology,\\ 
Institute of Quantum Matter, South China Normal University, Guangzhou 510006, Guangdong, China.
\\
$^c$
Guangdong Basic Research Center of Excellence for Structure and Fundamental Interactions of Matter, 
Guangdong Provincial Key Laboratory of Nuclear Science, Guangzhou 510006, Guangdong, China. 
}\\
\vskip 1mm
\vspace{40pt}
\end{center}

%\newpage
\begin{abstract}
\noindent
We introduce a chemical potential for non-Hermitian lattice fermions and show that, for even flavors with degenerate masses and paired chemical potentials $(\mu,-\mu)$ or $(i\mu,i\mu)$, the Hybrid Monte Carlo algorithm is free of the sign problem. 
For one-dimensional free fermions, we demonstrate that the sign problem is a numerical rather than physical obstruction and derive the exact propagator, which is analytic at finite lattice spacing away from its poles but becomes non-analytic in the continuum limit. Finally, we use AI-assisted fitting to perform analytic continuation from imaginary to real chemical potentials.

\end{abstract}
\end{titlepage}
\setcounter{tocdepth}{3}
{\hypersetup{linkcolor=black}\tableofcontents}
\newpage
\section{Introduction}
\label{sec:1}
\noindent
Lattice field theory introduces a finite lattice spacing to regularize the ultraviolet divergences that arise in the continuum path-integral formulation of quantum field theory (QFT) \cite{Feynman:1948ur,Wilson:1974sk,Makeenko:2002uj,Gattringer:2010zz}.
It provides a well-defined non-perturbative regularization of QFT. 
It offers a powerful framework for studying strongly coupled systems \cite{Ma:2024zbl}.
For consistency, the lattice theory is expected to reproduce the desired continuum QFT in the continuum limit, where the lattice spacing vanishes and the lattice volume becomes infinite.
However, establishing this connection is highly nontrivial.
For lattice fermions, non-perturbative lattice-spacing effects, such as those underlying the fermion doubling problem, must be properly taken into account rather than treated as small perturbative corrections \cite{Nielsen:1980rz,Nielsen:1981xu,Karsten:1981gd}.
Furthermore, the Euclidean lattice theory does not generally admit a straightforward Wick rotation \cite{Wick:1954eu} to its Lorentzian counterpart, making the analytic continuation itself a subtle issue \cite{Kosyakov:2026nrs}.
These challenges complicate the extraction of continuum QFT from lattice formulations.
\\

\noindent
The fermion-doubling theorem states that a local, Hermitian lattice fermion theory with chiral symmetry inevitably contains additional unwanted fermion species (doublers) in the continuum limit \cite{Nielsen:1980rz,Nielsen:1981xu,Karsten:1981gd}.
A standard solution is to introduce the Wilson term, which explicitly breaks chiral symmetry and gives the doublers masses of order the inverse lattice spacing, causing them to decouple in the continuum limit \cite{Kogut:1974ag}.
Alternatively, one can preserve a modified lattice chiral symmetry through the Ginsparg--Wilson relation \cite{Ginsparg:1981bj}, whose overlap fermion solution avoids doubling while maintaining an exact lattice chiral symmetry \cite{Luscher:1998pqa,Neuberger:1998wv}.
However, the overlap operator is intrinsically non-local, making numerical simulations computationally expensive.
Another approach employs one-sided lattice derivatives, which eliminate doublers at the cost of Hermiticity and hypercubic symmetry \cite{Stamatescu:1993ga,Stamatescu:1994yj}.
The breaking of hypercubic symmetry leads to divergent continuum observables \cite{Sadooghi:1996ip}.
This problem can be remedied by quenched averaging over all lattice directions, which restores hypercubic symmetry in physical observables and yields the correct continuum limit \cite{Stamatescu:1993ga,Stamatescu:1994yj}.
Although non-Hermiticity appears to preclude conventional Monte Carlo simulations, pairing forward and backward lattice derivatives produces a non-negative fermion determinant, allowing the introduction of pseudofermions and {\it enabling} Monte Carlo algorithms for non-Hermitian lattice fermions \cite{Guo:2021sjp,Guo:2024jqt}.
\\

\noindent
In continuum quantum field theory (QFT), Lorentzian correlation functions are commonly obtained from the Euclidean theory through Wick rotation \cite{Schwinger:1958qau,Osterwalder:1974tc}.
On the lattice, however, this correspondence is not generic \cite{Kosyakov:2026nrs}.
For lattice scalar fields, the pole structure of the Green's function can change under Wick rotation, preventing analytic continuation before the continuum limit is taken \cite{Kosyakov:2026nrs,Schwinger:1951ex}.
Although the Euclidean lattice theory possesses a well-defined path integration measure, the failure of analytic continuation explains why an equally well-defined Lorentzian measure is generally absent \cite{Kosyakov:2026nrs}.
Consequently, analytic continuation and the continuum limit do not generally commute, leading to inequivalent theories \cite{Kosyakov:2026nrs}.
Despite these subtleties, analytic continuation remains an {\it indispensable} tool for investigating non-perturbative phenomena that are otherwise inaccessible \cite{Lombardo:1999cz,DElia:2002tig,Braun:2012ww}.
A notable example is the search for the QCD critical endpoint, where lattice simulations are often performed at imaginary chemical potential to {\it avoid} the sign problem and then {\it analytically continued} to real chemical potential \cite{Lombardo:1999cz,DElia:2002tig,Braun:2012ww,Brandt:2017oyy,Hasenfratz:1983ba}.
The reliability of this approach depends crucially on the {\it analyticity} of the partition function and observables as functions of the chemical potential, which is generally difficult to establish \cite{Lombardo:1999cz,DElia:2002tig,Brandt:2017oyy}.
Consequently, analytic continuation alone {\it cannot} yet provide a definitive determination of the QCD phase structure.
Since the one-dimensional free lattice fermion with a non-Hermitian discretization admits an {\it exact} solution, it provides an ideal setting for examining the validity of analytic continuation explicitly.
Understanding when analytic continuation is justified is essential for exploring non-perturbative regimes that remain inaccessible to direct numerical methods, with implications extending to a wide range of quantum field theories and related areas of fundamental physics \cite{Giveon:1994fu,Ma:2018efs,Ma:2023krt,Ma:2025zaz}.

\subsection{Summary of Results}
\noindent
In this paper, we investigate analytic continuation with respect to the chemical potential in lattice field theory.
Finite-density lattice simulations generally suffer from the sign problem, making analytic continuation from imaginary to real chemical potentials one of the few practical approaches for exploring the phase diagram.
However, the validity of this procedure relies on the analyticity of the underlying observables, which is difficult to establish in general. To address this issue, we study a one-dimensional free Dirac fermion on a lattice, for which all relevant quantities can be computed exactly.
This exactly solvable model enables us to explicitly examine the validity of analytic continuation and clarify its behavior in finite-density lattice field theory. 
Our main results are summarized as follows:
\begin{itemize}
\item We derive the exact two-point fermion correlation functions with both the naive and exponential implementations of the chemical potential.
We show that, at finite lattice spacing, analytic continuation can be performed.
At the same time, taking the infinite-volume limit can destroy analyticity.
Consequently, the continuum limit and analytic continuation do not generally commute.
Moreover, for the exponential implementation, the exact solution shows that the chemical potential does not affect the continuum limit, providing further justification for adopting the exponential prescription on the lattice.

\item We extend the Hybrid Monte Carlo (HMC) algorithm for non-Hermitian lattice fermions to finite chemical potential by considering degenerate fermion masses and the flavor pairs $(\mu_1,\mu_2)=(\mu,-\mu)$ and $(i\mu, i\mu)$, where $\mu$ is real. We compute the two-point pseudofermion correlation functions and find good agreement between the HMC simulations and the exact solutions for lattice sizes $N_t=16$ and $32$.

\item We analyze the eigenvalue spectrum of the Dirac operator and show that the sign problem is a technical obstruction to introducing pseudofermions rather than an indication of an ill-defined fermionic theory.
For two degenerate flavors, the configurations studied here show a correspondence between the sign problem and eigenvalues with negative real parts.

\item We perform AI-assisted analytic continuation from imaginary to real chemical potentials using the Laurent exponential model and the physics-constrained neural network (PCNN).
The reconstructed correlation functions agree well with the exact results while requiring only a small amount of training data, thereby providing a benchmark for the method across different lattice sizes.
\end{itemize}

\noindent
The remainder of this paper is organized as follows.
In Sec.~\ref{sec:2}, we derive the exact solution for non-Hermitian lattice fermions with chemical potential.
In Sec.~\ref{sec:3}, we formulate the HMC algorithm for degenerate fermions with specific chemical-potential assignments.
In Sec.~\ref{sec:4}, we analyze the eigenvalue spectrum of the Dirac operator and discuss its relation to the sign problem.
In Sec.~\ref{sec:5}, we compare the HMC results for the two-point pseudofermion correlation functions with the exact solutions.
In Sec.~\ref{sec:6}, we present the AI-assisted analytic continuation based on the Laurent exponential model and the PCNN.
Finally, Sec.~\ref{sec:7} contains our conclusions and outlook.

%The main results of our work are presented in Fig. \ref{summary.png}.
%\begin{figure}[tbp]
%\centering
%\includegraphics[scale = 1]{fig-Z23.pdf}
%\caption{Graphical notation for the tripartite wavefunction $\Psi$, its conjugate $\Psi^*$, the reduced density matrix $\rho_{BC}$, and the construction of $\mathcal{Z}^{(\mathtt{q})}_{2}$ in the special case $\mathtt{q} = 3$. }
%\label{fig:Z23}
%\end{figure}

\section{Non-Hermitian Lattice Fermion with Chemical Potential}
\label{sec:2}
\noindent
We begin with the 1D free Dirac fermion theory at finite chemical potential.
We then introduce the non-Hermitian lattice formulation \cite{Stamatescu:1993ga} and implement the chemical potential in both the naive and exponential forms \cite{Hasenfratz:1983ba}.
The exact propagator shows that, on a finite lattice, analytic continuation is valid.
In the continuum limit, however, the analyticity domain becomes restricted by the limiting pole structure.

\subsection{Free Dirac Fermion}
\noindent
We introduce the 1D free Dirac fermion theory with a chemical potential
\bea
S_{\mathrm{1D}}\lbrack\bar{\psi}, \psi\rbrack
=
\int dx\ \big(\bar{\psi}(\fsl{\partial}+m-\mu_1\gamma_1)\psi\big),
\eea
where $m$ is the Dirac fermion mass, and $\mu_1$ is the chemical potential.
Our convention for the $\gamma$ matrix is:
\bea
\gamma_1=\sigma_z=\begin{pmatrix}
1&0
\\
0&-1
\end{pmatrix}.
\eea
The propagator is
\bea
S(x)=\int^{\infty}_{-\infty}\frac{dp}{2\pi}\frac{e^{ipx}}{i\gamma_1p+m-\gamma_1 \mu_1}.
\eea
The propagator has two diagonal elements, while the off-diagonal elements vanish.
The first diagonal element of the propagator is
\bea
S_+(x)=\int^{\infty}_{-\infty}\frac{dp}{2\pi}\ \frac{e^{ipx}}{ip+m-\mu_1}
=\Bigg\{\begin{array}{ll}
\theta(x)e^{-(m-\mu_1)x}, & \mathrm{Re}(m-\mu_1)>0 \\
-\theta(-x)e^{-(m-\mu_1)x}, & \mathrm{Re}(m-\mu_1)<0
\end{array},
\eea
where
\bea
\theta(x)\equiv\Bigg\{\begin{array}{ll}
1, & x\ge 0 \\
0, & x<0
\end{array}.
\eea
The second diagonal element is
\bea
S_-(x)=\int^{\infty}_{-\infty}\frac{dp}{2\pi}\ \frac{e^{ipx}}{-ip+m+\mu_1}
=\Bigg\{\begin{array}{ll}
\theta(-x)e^{(m+\mu_1)x}, & \mathrm{Re}(m+\mu_1)>0 \\
-\theta(x)e^{(m+\mu_1)x}, & \mathrm{Re}(m+\mu_1)<0
\end{array}.
\eea
For a purely imaginary chemical potential, the exact solution lies in the region $\mathrm{Re}(m\mp\mu_1)>0$.
When the real chemical potential is continued beyond this region, i.e., when $\mathrm{Re}(m\mp\mu_1)<0$, the analytic continuation from imaginary to real chemical potential fails.
Thus, the analytic continuation is valid only in the region
\bea
m>\mathrm{Re}(\mu_1)>-m.
\eea

\subsection{Non-Hermitian Lattice Formulation}
\noindent
We now consider the lattice theory obtained from the forward finite-difference
\bea
S_{\mathrm{F}}&=&a\sum_{n=0}^{N_t-1}\bar{\psi}(n)\bigg(\gamma_1\frac{\psi(n+1)-\psi(n)}{a}+m\psi(n)-\mu_1\gamma_1\psi(n)\bigg)
\nn\\
&=&a\sum_{n_1, n_2; \alpha_1, \alpha_2}\bar{\psi}(n_1)_{\alpha_1}\big(D(n_1, n_2)_{\alpha_1, \alpha_2}+m\delta_{n_1, n_2}\delta_{\alpha_1, \alpha_2}
-\mu_1(\gamma_1)_{\alpha_1, \alpha_2}\delta_{n_1, n_2}
\big)\psi(n_2)_{\alpha_2},
\nn\\
\eea
where
\bea
D(n_1, n_2)_{\alpha_1, \alpha_2}\equiv
(\gamma_1)_{\alpha_1, \alpha_2}\frac{\delta_{n_1+1, n_2}-\delta_{n_1, n_2}}{a}.
\eea
We label the matrix components of $\gamma_1$ by $\alpha_1, \alpha_2=1, 2$.
$N_t$ is the number of lattice points.
The lattice fermion field satisfies the anti-periodic boundary condition
\bea
\psi(0)=-\psi(N_t).
\eea
\\

\noindent
For $x\equiv na>0$, the lattice propagator is
\bea
&&
S_{L}(x)
\nn\\
&=&\frac{1}{N_t}\sum_{j=0}^{N_t-1}\exp\bigg(i\frac{(2j+1)\pi}{N_t}n\bigg)
\begin{pmatrix}
\frac{1}{\exp\big(i\frac{(2j+1)\pi}{N_t}\big)-1+(m-\mu_1)a}&0
\\
0&\frac{1}{-\exp\big(i\frac{(2j+1)\pi}{N_t}\big)+1+(m+\mu_1)a}
\end{pmatrix}
\nn\\
&=&
\frac{1}{N_t}\oint_{C_1} dw\ \exp\bigg(i\frac{2w\pi}{N_t}n\bigg)
\begin{pmatrix}
\frac{\frac{-1}{\exp(2\pi i w)+1}}{\exp\big(i\frac{2w\pi}{N_t}\big)-1+(m-\mu_1)a}&0
\\
0&\frac{\frac{-1}{\exp(2\pi i w)+1}}{-\exp\big(i\frac{2w\pi}{N_t}\big)+1+(m+\mu_1)a}
\end{pmatrix}
\nn\\
&=&
\begin{pmatrix}
\frac{\big(1-(m-\mu_1)a\big)^{n-1}}{\big(1-(m-\mu_1)a\big)^{N_t}+1}& 0
\\
0& -\frac{\big(1+(m+\mu_1)a\big)^{n-1}}{\big(1+(m+\mu_1)a\big)^{N_t}+1}
\end{pmatrix}.
\eea
The closed contour $C_1$ encloses the poles $w=1/2, 3/2 \cdots, N_t-1/2$; see Fig. \ref{closed_loop_C1}.
For $x<0$, we replace $\exp(2\pi i w)$  with $\exp(-2\pi i w)$ to avoid a divergent boundary contribution.
The lattice propagator becomes:
\bea
S_{L}(x)
&=&
\frac{1}{N_t}\oint_{C_1} dw\ \exp\bigg(i\frac{2w\pi}{N_t}n\bigg)
\begin{pmatrix}
\frac{\frac{1}{\exp(-2\pi i w)+1}}{\exp\big(i\frac{2w\pi}{N_t}\big)-1+(m-\mu_1)a}&0
\\
0&\frac{\frac{1}{\exp(-2\pi i w)+1}}{-\exp\big(i\frac{2w\pi}{N_t}\big)+1+(m+\mu_1)a}
\end{pmatrix}
\nn\\
&=&
\begin{pmatrix}
-\frac{\big(1-(m-\mu_1)a\big)^{n-1}}{\big(1-(m-\mu_1)a\big)^{-N_t}+1}& 0
\\
0& \frac{\big(1+(m+\mu_1)a\big)^{n-1}}{\big(1+(m+\mu_1)a\big)^{-N_t}+1}
\end{pmatrix}.
\eea
In the infinite-volume limit, for sufficiently small lattice spacing and $\mathrm{Re}(m\mp\mu_1)>0$, the lattice propagator becomes
\bea
S_{+}(x)&\rightarrow&\Bigg\{\begin{array}{ll}
\big(1-(m-\mu_1)a\big)^{n-1}, & x> 0 \\
0, & x<0
\end{array};
\nn\\
S_{-}(x)&\rightarrow&\Bigg\{\begin{array}{ll}
0, & x> 0 \\
\big(1+(m+\mu_1)a\big)^{n-1}, & x<0
\end{array}.
\eea
For $\mathrm{Re}(m\mp\mu_1)<0$, it becomes
\bea
S_{+}(x)&\rightarrow&\Bigg\{\begin{array}{ll}
0, & x> 0 \\
-\big(1+(m-\mu_1)a\big)^{n-1}, & x<0
\end{array};
\nn\\
S_{-}(x)&\rightarrow&\Bigg\{\begin{array}{ll}
-\big(1-(m+\mu_1)a\big)^{n-1}, & x> 0 \\
0, & x<0
\end{array},
\eea
The non-physical contribution vanishes in the infinite-volume limit.
We then take the continuum limit
\bea
(m-\mu_1)a\rightarrow 0; \ \frac{a}{x}\rightarrow 0.
\eea
After taking the continuum limit, the lattice propagator reduces to that of a 1D free Dirac fermion.
Our lattice solution shows that the propagator has the same analytic form for any value of $\mu_1$ on a finite lattice.
Therefore, analytic continuation works on a finite lattice.
However, this property is lost in the infinite-volume limit.
\begin{figure}[tbp]
\centering
\includegraphics[scale = 0.28]{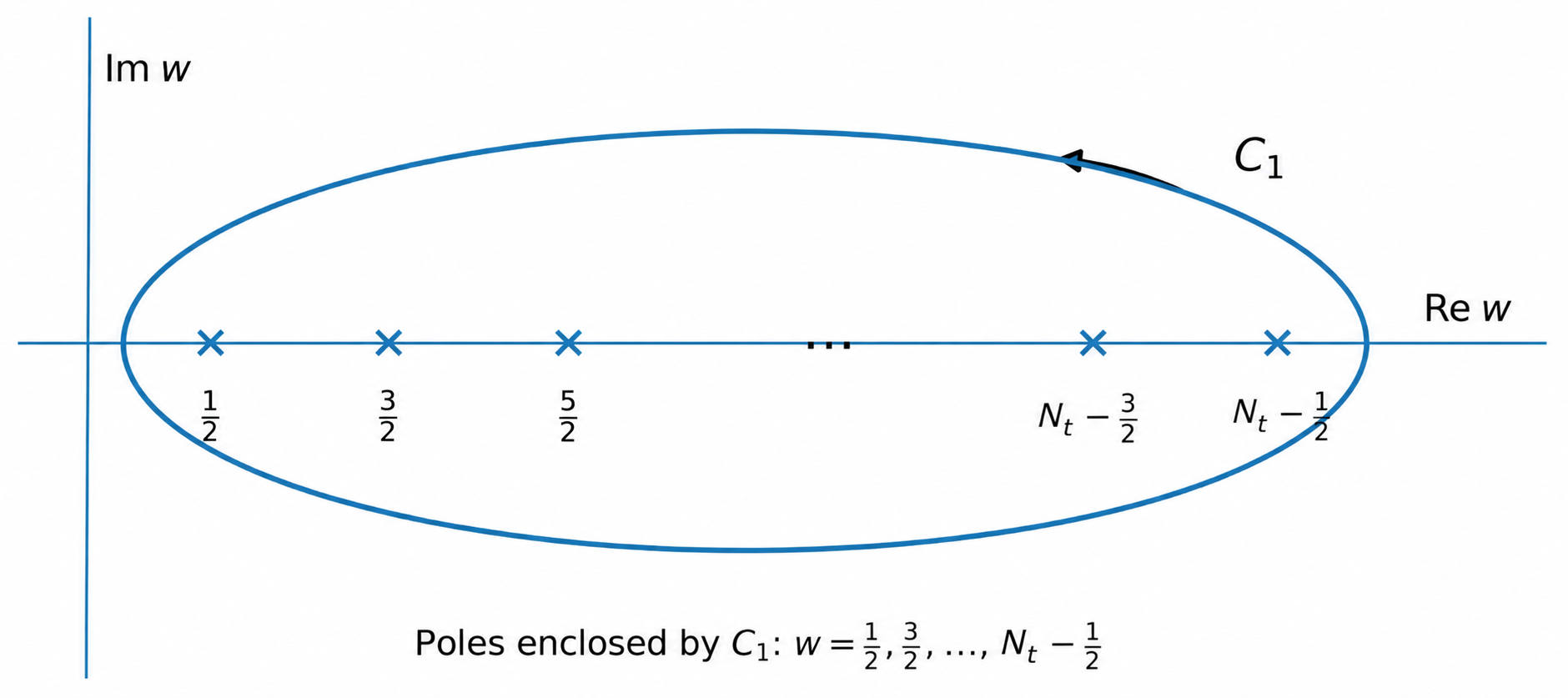}
\caption{Closed contour $C_1$ in the complex $w$-plane.
The contour is oriented counterclockwise and encloses the $N_t$ simple poles located at $w=(2j+1)/2$ with $j=0, 1, \cdots, N_t-1$, namely $w=1/2, 3/2, \cdots, N_t-1/2$.
}
\label{closed_loop_C1}
\end{figure}

\subsection{Exponential Form}
\noindent
We now consider the lattice theory obtained from the forward finite-difference
\bea
S_{\mathrm{F}}&=&a\sum_{n=0}^{N_t-1}\bar{\psi}(n)\bigg(\gamma_1\frac{e^{-a\mu_1}\psi(n+1)-\psi(n)}{a}+m\psi(n)\bigg)
\nn\\
&=&a\sum_{n_1, n_2; \alpha_1, \alpha_2}\bar{\psi}(n_1)_{\alpha_1}\big(D(\mu_1, n_1, n_2)_{\alpha_1, \alpha_2}+m\delta_{n_1, n_2}\delta_{\alpha_1, \alpha_2}
\big)\psi(n_2)_{\alpha_2},
\eea
where
\bea
D(\mu_1, n_1, n_2)_{\alpha_1, \alpha_2}\equiv
(\gamma_1)_{\alpha_1, \alpha_2}\frac{e^{-a\mu_1}\delta_{n_1+1, n_2}-\delta_{n_1, n_2}}{a}.
\eea
For $x\equiv na>0$, the lattice propagator is
\bea
&&
S_{L}(x)
\nn\\
&=&\frac{1}{N_t}\sum_{j=0}^{N_t-1}\exp\bigg(i\frac{(2j+1)\pi}{N_t}n\bigg)
\nn\\
&&\times
\begin{pmatrix}
\frac{1}{\exp(-\mu_1 a)\exp\big(i\frac{(2j+1)\pi}{N_t}\big)-1+ma}&0
\\
0&\frac{1}{-\exp(-\mu_1 a)\exp\big(i\frac{(2j+1)\pi}{N_t}\big)+1+ma}
\end{pmatrix}
\nn\\
&=&
\frac{1}{N_t}\oint_{C_1} dw\ \exp\bigg(i\frac{2w\pi}{N_t}n\bigg)
\begin{pmatrix}
\frac{\frac{-1}{\exp(2\pi i w)+1}}{\exp(-\mu_1 a)\exp\big(i\frac{2w\pi}{N_t}\big)-1+ma}&0
\\
0&\frac{\frac{-1}{\exp(2\pi i w)+1}}{-\exp(-\mu_1 a)\exp\big(i\frac{2w\pi}{N_t}\big)+1+ma}
\end{pmatrix}
\nn\\
&=&e^{n\mu a}
\begin{pmatrix}
\frac{(1-ma)^{n-1}}{e^{N_t\mu_1 a}(1-ma)^{N_t}+1}& 0
\\
0& -\frac{(1+ma)^{n-1}}{e^{N_t\mu_1 a}(1+ma)^{N_t}+1}
\end{pmatrix}.
\eea
The closed contour $C_1$ also encloses the poles $w=1/2, 3/2 \cdots, N_t-1/2$.
\\

\noindent
For $x<0$, we again replace $\exp(2\pi i w)$  with $\exp(-2\pi i w)$ to avoid a divergent boundary contribution.
The lattice propagator becomes:
\bea
&&
S_{L}(x)
\nn\\
&=&
\frac{1}{N_t}\oint_{C_1} dw\ \exp\bigg(i\frac{2w\pi}{N_t}n\bigg)
\begin{pmatrix}
\frac{\frac{1}{\exp(-2\pi i w)+1}}{\exp(-\mu_1 a)\exp\big(i\frac{2w\pi}{N_t}\big)-1+ma}&0
\\
0&\frac{\frac{1}{\exp(-2\pi i w)+1}}{-\exp(-\mu_1 a)\exp\big(i\frac{2w\pi}{N_t}\big)+1+ma}
\end{pmatrix}
\nn\\
&=&e^{n\mu_1 a}
\begin{pmatrix}
-\frac{(1-ma)^{n-1}}{e^{-N_t\mu_1 a}(1-ma)^{-N_t}+1}& 0
\\
0& \frac{(1+ma)^{n-1}}{e^{-N_t\mu_1 a}(1+ma)^{-N_t}+1}
\end{pmatrix}.
\eea
In the infinite-volume limit and for infinitesimal lattice spacing, with $\mathrm{Re}(m\mp\mu_1)>0$, the lattice propagator becomes
\bea
S_{+}(x)&\rightarrow&\Bigg\{\begin{array}{ll}
e^{-(m-\mu_1)x}, & x> 0 \\
0, & x<0
\end{array};
\nn\\
S_{-}(x)&\rightarrow&\Bigg\{\begin{array}{ll}
0, & x> 0 \\
e^{(m+\mu_1)x}, & x<0
\end{array},
\eea
For $\mathrm{Re}(m\mp\mu_1)<0$, it becomes
\bea
S_{+}(x)&\rightarrow&\Bigg\{\begin{array}{ll}
0, & x> 0 \\
-e^{-(m-\mu_1)x}, & x<0
\end{array};
\nn\\
S_{-}(x)&\rightarrow&\Bigg\{\begin{array}{ll}
-e^{(m+\mu_1)x}, & x> 0 \\
0, & x<0
\end{array},
\eea
After taking the continuum limit, the lattice propagator reduces to that of the 1D free Dirac fermion theory.
Our lattice solution again shows that analytic continuation is valid.
The infinite-lattice limit breaks analyticity.
Because the chemical potential is introduced in exponential form, reducing the continuum field theory requires only
\bea
ma\rightarrow 0; \ \frac{a}{x}\rightarrow 0.
\eea
Thus, no additional constraint on the chemical potential is required in the continuum limit.
The exponential prescription is also convenient when gauge fields are included in a consistent lattice formulation.
We therefore use the exponential form for the chemical potential on a finite lattice \cite{Hasenfratz:1983ba}.

\section{Implementation of Monte Carlo Algorithm}
\label{sec:3}
\noindent
Although we consider a single Dirac fermion, the non-Hermiticity of the lattice formulation \cite{Stamatescu:1993ga} prevents a direct Monte Carlo treatment.
We therefore consider an even number of flavors with degenerate mass \cite{Guo:2021sjp}.
The zero-chemical-potential implementation is already established \cite{Guo:2021sjp}, and we generalize the same approach to nonzero chemical potential.
\\

\noindent
For two fermion fields in 1D with a degenerate mass, the lattice action is
\bea
&&
S_{FD}
\nn\\
&=&a\sum_{n_1, n_2=0}^{N_t-1}\bigg(\bar{\psi}_1(n_1)\big(D_+(\mu_1, n_1, n_2)+m\delta_{n_1, n_2}\big)\psi_1(n_2)
\nn\\
&&+
\bar{\psi}_2(n_1)\big(D_-(\mu,_2, n_1, n_2)+m\delta_{n_1, n_2}\big)\psi_2(n_2)\bigg).
\eea
Here we adopt the forward finite-difference for $\psi_1$,
\bea
D_+(\mu_1, n_1, n_2)_{\alpha_1, \alpha_2}\equiv
(\gamma_1)_{\alpha_1, \alpha_2}\frac{e^{-a\mu_1}\delta_{n_1+1, n_2}-\delta_{n_1, n_2}}{a}.
\eea
For $\psi_2$, we use the backward finite difference with the same accuracy,
\bea
D_-(\mu_2, n_1, n_2)_{\alpha_1, \alpha_2}\equiv
(\gamma_1)_{\alpha_1, \alpha_2}\frac{\delta_{n_1, n_2}-e^{a\mu_2}\delta_{n_1-1, n_2}}{a}.
\eea
At zero chemical potential, the backward finite-difference can be written as
\bea
D_-(0, n_1, n_2)=-\big(D_+(0, n_1, n_2)\big)^{\dagger}.
\label{BF}
\eea
We specifically consider ($\mu_1$, $\mu_2$)=($\mu$, $-\mu$) and ($i\mu$, $i\mu$), where $\mu$ is real.
For these choices, the backward difference satisfies the analog of Eq. \eqref{BF}
\bea
D_-(\mu_2, n_1, n_2)=-\big(D_+(\mu_1, n_1, n_2)\big)^{\dagger}.
\eea
After integrating out the fermion fields, we obtain a non-negative determinant:
\bea
&&
\det\big(D_+(\mu_1, n_1, n_2)+m\big)\det\bigg(-\big(D_+(\mu_1, n_1, n_2)\big)^{\dagger}+m\bigg)
\nn\\
&=&\det\big(D_+(\mu_1, n_1, n_2)+m\big)\det\Bigg(\gamma_5\bigg(-\big(D_+(\mu_1, n_1, n_2)\big)^{\dagger}+m\bigg)\gamma_5\Bigg)
\nn\\
&=&\big|\det\big(D_+(\mu_1, n_1, n_2)+m\big)\big|^2.
\eea
Introducing the pseudo-fermion field $\phi_f$ allows us to rewrite the partition function as
\bea
&&
\int {\cal D}\bar{\psi}{\cal D}\psi\ \exp(-S_{FD})
\sim
\int {\cal D}\phi_{f, R}{\cal D}\phi_{f, I}\ \exp\bigg(-\phi_f^{\dagger}\big((D_++m)(D_+^{\dagger}+m)\big)^{-1}\phi_f\bigg),
\nn\\
\eea
where
\bea
\phi_f\equiv \phi_{f, R}+i\phi_{f, I}.
\eea
Thus, the pseudo-fermion action is Hermitian.
The Monte Carlo algorithm can be applied to the non-Hermitian lattice formulation because the Hermiticity of the original Dirac operator is not required; only a non-negative partition function is required.

\section{Eigenvalues of Dirac Matrices}
\label{sec:4}
\noindent
For the Gaussian path integral over pseudo-fermions to converge in the standard representation, the relevant quadratic operator must have eigenvalues with non-negative real parts.
The Dirac operator itself, however, need not satisfy this condition: its square is related to the Klein-Gordon operator, and negative real parts of Dirac eigenvalues do not by themselves imply an ill-defined fermionic theory.
The eigenvalue results show negative real parts for $N_t=16$ and $32$, as shown in Figs. \ref{eigenvalues1} and \ref{eigenvalues2}.
Because both the lattice propagator and its continuum limit are well defined, these negative real parts should not be interpreted as evidence that the fermionic theory is physically pathological.
For the even-flavor theory, we compute the eigenvalues relevant to the pseudo-fermion path integral.
When the chemical potentials are $(\mu_1, \mu_2)=(\mu, \mu), (i\mu, -i\mu)$, the simulation suffers from the sign problem.
For the other chemical-potential assignments, $(\mu_1, \mu_2)=(\mu, -\mu), (i\mu, i\mu)$, the simulation is free of the sign problem.
%These four considerations are given in Tab. \ref{sign}.
The eigenvalue results (Figs. \ref{eigenvalues3} and \ref{eigenvalues4}) show that, for the configurations studied here, the occurrence of a sign problem coincides with the presence of an eigenvalue with a negative real part.
\begin{figure}[tbp]
\centering
\includegraphics[scale = 1.2]{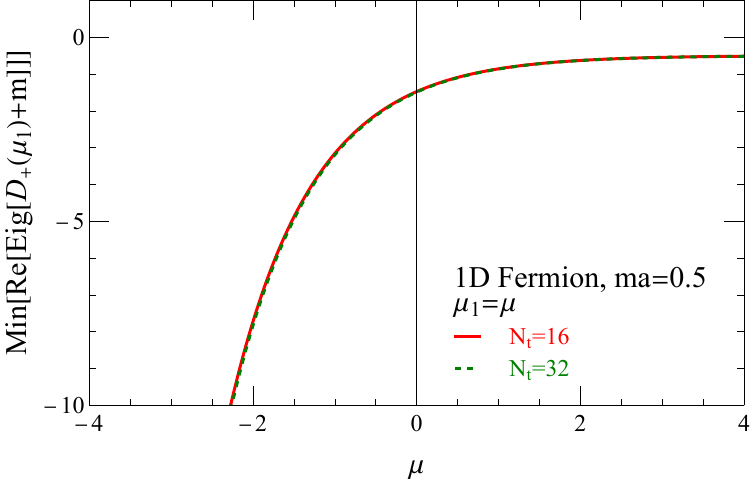}
\includegraphics[scale = 1.2]{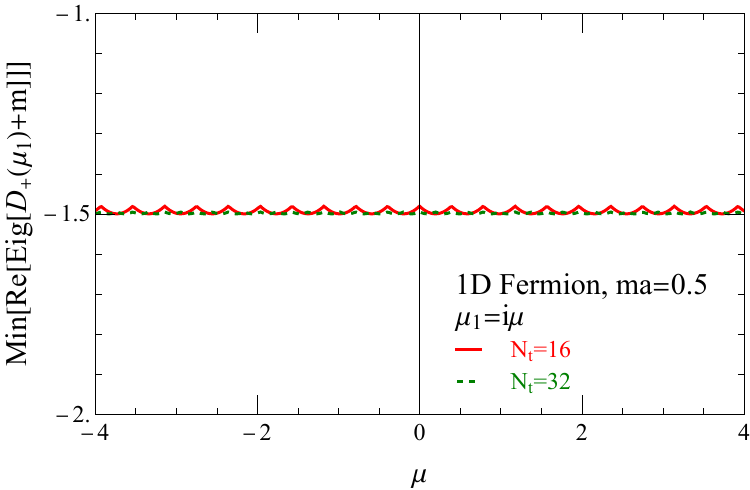}
\caption{The smallest real part of eigenvalues of a Dirac matrix corresponding to the forward derivative is not positive definite for $N_t=16, 32$ with $ma=0.5$.
}
\label{eigenvalues1}
\end{figure}
\begin{figure}[tbp]
\centering
\includegraphics[scale = 1.2]{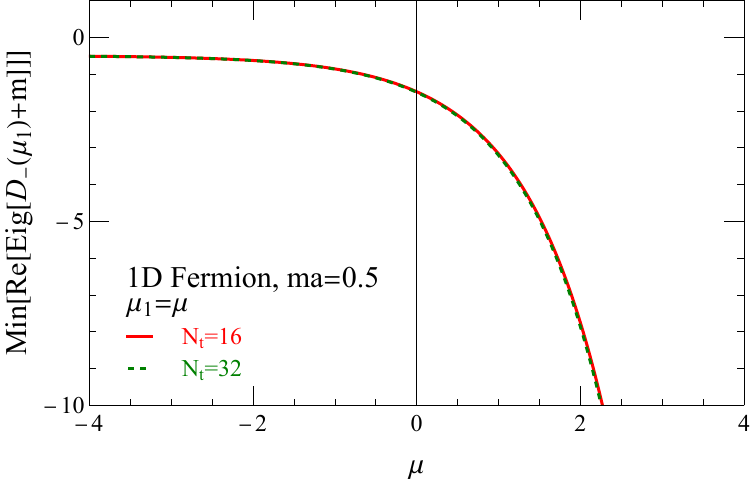}
\includegraphics[scale = 1.2]{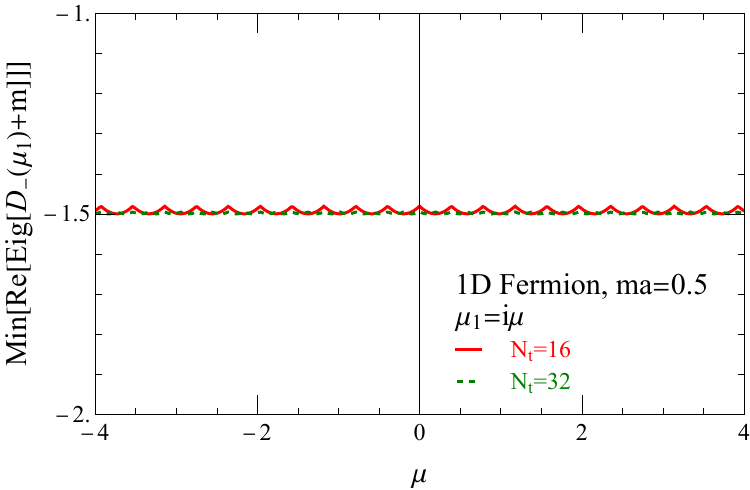}
\caption{The smallest real part of eigenvalues of a Dirac matrix corresponding to the backward derivative is not positive definite for $N_t=16, 32$ with $ma=0.5$.
}
\label{eigenvalues2}
\end{figure}
%\begin{table}[tbp]
%\centering
%\begin{tabular}{|c|c|c|c|c|}
%\hline\hline
%$\mu_1$             & $(+,\operatorname{Re})$ & $(+,\operatorname{Re})$ & $(+,\operatorname{Im})$ & $(+,\operatorname{Im})$ \\
%$\mu_2$             & $(+,\operatorname{Re})$ & $(-,\operatorname{Re})$ & $(+,\operatorname{Im})$ & $(-,\operatorname{Im})$ \\
%\hline
%Sign Problem & Yes      & No       & No       & Yes \\
%\hline\hline
%\end{tabular}
%\caption{The $+$ and $-$ denote positive and negative signs for the chemical potentials $(\mu_1, \mu_2)$, while $\mathrm{Re}$ and $\mathrm{Im}$ denote real and imaginary quantities, respectively.
%The "Yes" and "No" denote whether the simulation suffers from the sign problem.
%}
%\label{sign}
%\end{table}
\begin{figure}[tbp]
\centering
\includegraphics[scale = 1.2]{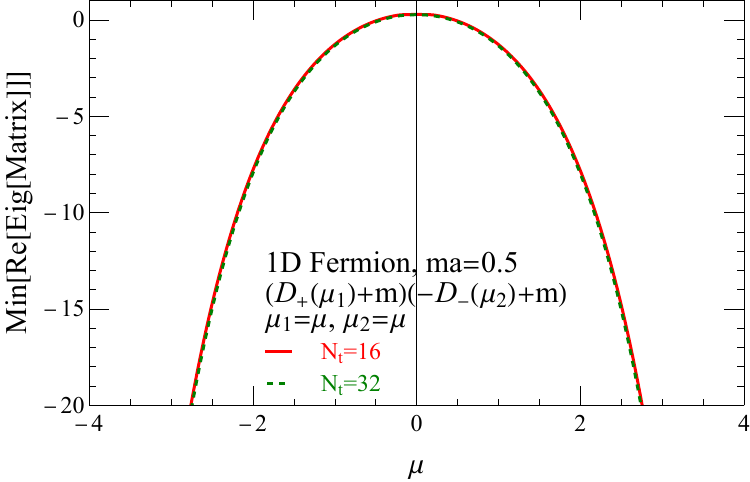}
\includegraphics[scale = 1.2]{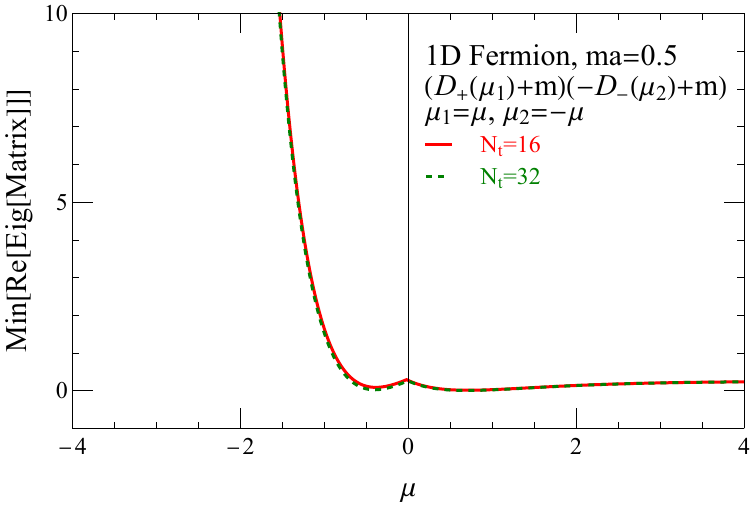}
\caption{
When the simulation suffers from the sign problem, the smallest real part of the eigenvalues of the Dirac matrices corresponding to the forward+backward derivative for real chemical potentials is not positive definite for $N_t=16, 32$ with $ma=0.5$.
}
\label{eigenvalues3}
\end{figure}
\begin{figure}[tbp]
\centering
\includegraphics[scale = 1.1]{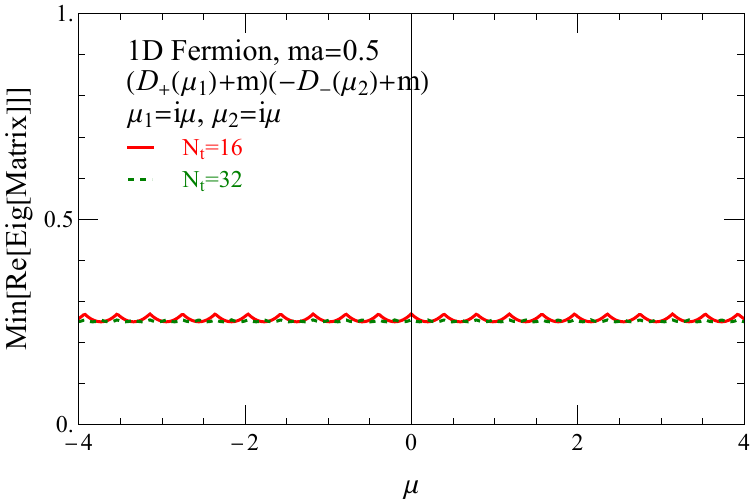}
\includegraphics[scale = 1.1]{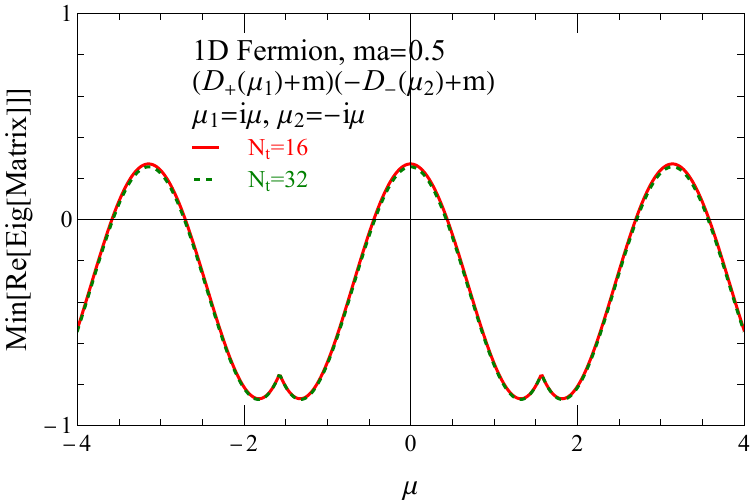}
\caption{
When the simulation suffers from the sign problem, the smallest real part of the eigenvalues of Dirac matrices corresponding to the forward+backward derivative for the imaginary chemical potentials is not always positive definite for $N_t=16, 32$ with $ma=0.5$.
}
\label{eigenvalues4}
\end{figure}

\newpage
\section{Lattice Simulation}
\label{sec:5}
\noindent
We implement the Hybrid Monte Carlo (HMC) algorithm for the chemical-potential assignments $(\mu_1, \mu_2)=(\mu, -\mu)$ and $(i\mu, i\mu)$ to calculate
\bea
{\cal O}_{jk}^{\alpha\beta}\equiv
\frac{1}{2}
\langle \phi^{\dagger \alpha}_{f, j}\phi^{\beta}_{f, k}
+\phi^{\dagger \beta}_{f, k}\phi^{\alpha}_{f,j}
\rangle=\big((D_++m)(D_++m)^{\dagger}\big)_{jk}^{\alpha\beta},
\eea
where $j, k=1, 2, \cdots, N_t$ and $\alpha, \beta=1, 2$.
We compare the exact solution with the numerical results for $N_t=16$ and $32$ in Figs. \ref{pf1}, \ref{pf2}, \ref{pf3}.
The agreement between the exact and numerical results indicates that lattice sizes $N_t=16$ and $32$ are sufficient to reproduce the behavior examined here. 
%Our analysis of the thermalization and autocorrelation time is given in Appendices \ref{sec:A} and \ref{sec:B}.
%According to the analysis, we consider thermalization for $2^6$ sweeps and measure intervals for $2^5$ sweeps.
%Including the chemical potentials in the simulation does not increase the required simulation resources for the thermalization and measurement intervals.
%Therefore, the non-Hermitian lattice formulation remains a practical method for lattice simulation.
\begin{figure}[tbp]
\centering
\includegraphics[scale = 0.4]{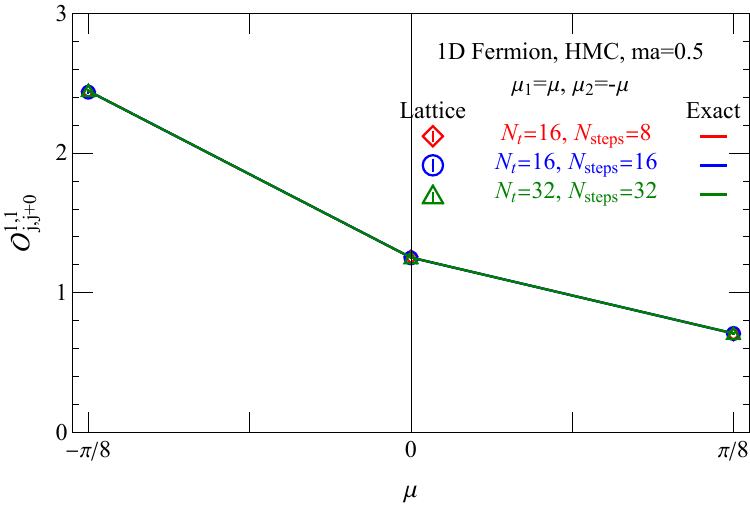}
\includegraphics[scale = 0.4]{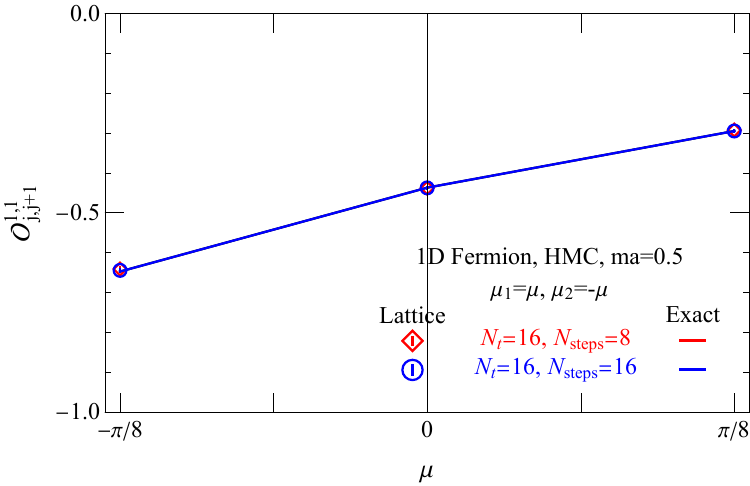}
\includegraphics[scale = 0.4]{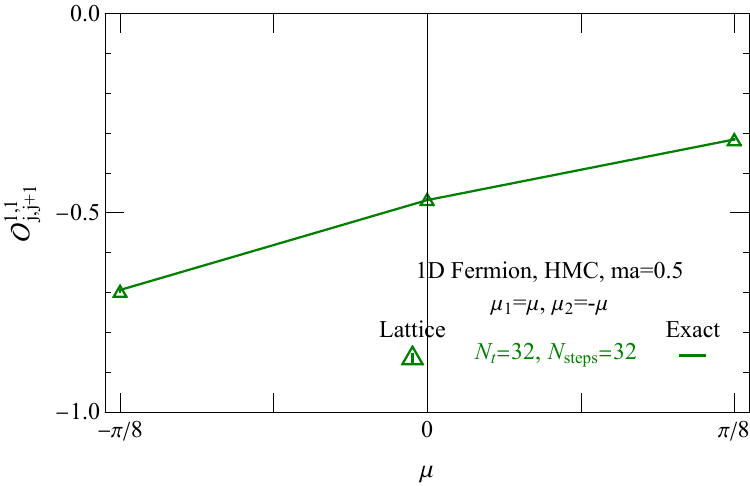}
\includegraphics[scale = 0.4]{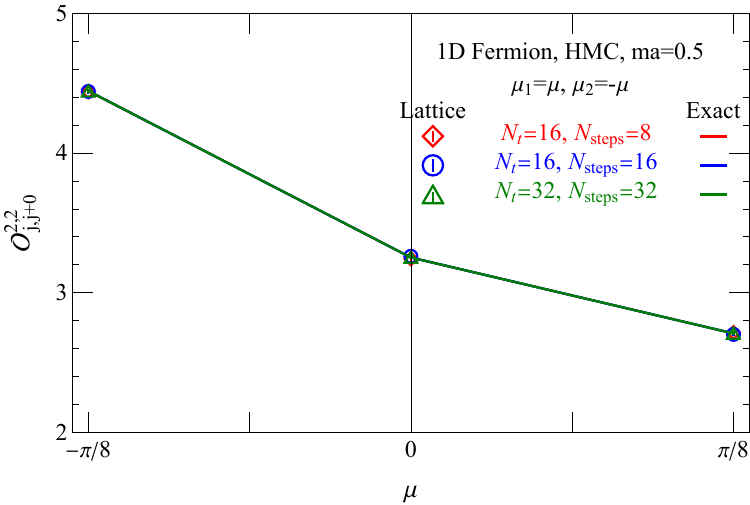}
\includegraphics[scale = 0.4]{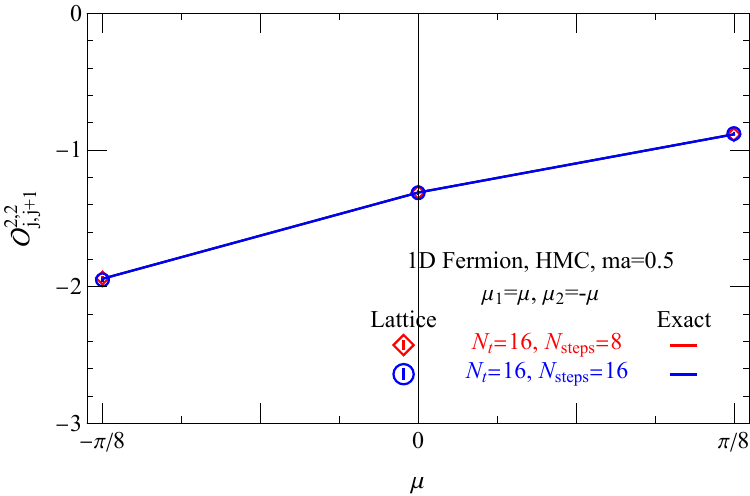}
\includegraphics[scale = 0.4]{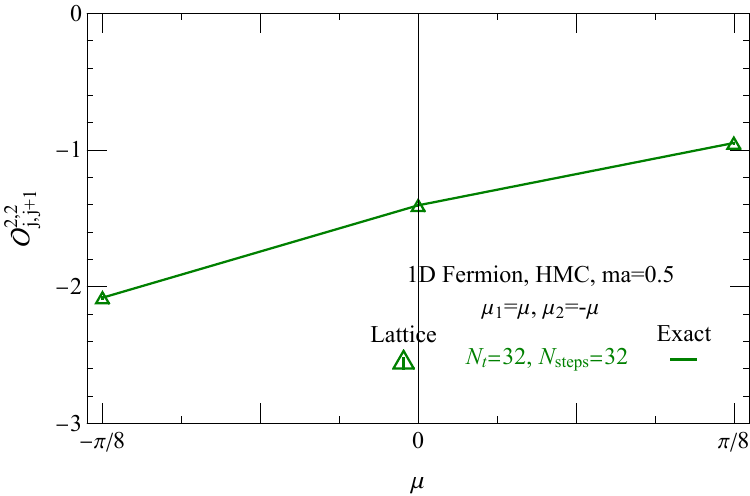}
\caption{
We use the Hybrid Monte Carlo (HMC) algorithm to compute $O^{\alpha_1,\alpha_2}_{j,j+n}$ for real chemical potentials $(\mu_1, \mu_2)=(\mu, -\mu)$ with $ma=0.5$, $(\alpha_1, \alpha_2)=(1, 1), (2, 2)$, and $n=0, 1$ for $N_t=16$ and $32$. 
The simulation uses $2^{14}$ measurement sweeps, $2^6$ thermalization sweeps, and a measurement interval of $2^5$ sweeps.
The error bars are less than $1\%$.
Here, $N_{\mathrm{steps}}$ denotes the number of molecular-dynamics steps.
}
\label{pf1}
\end{figure}
\begin{figure}[tbp]
\centering
\includegraphics[scale = 0.4]{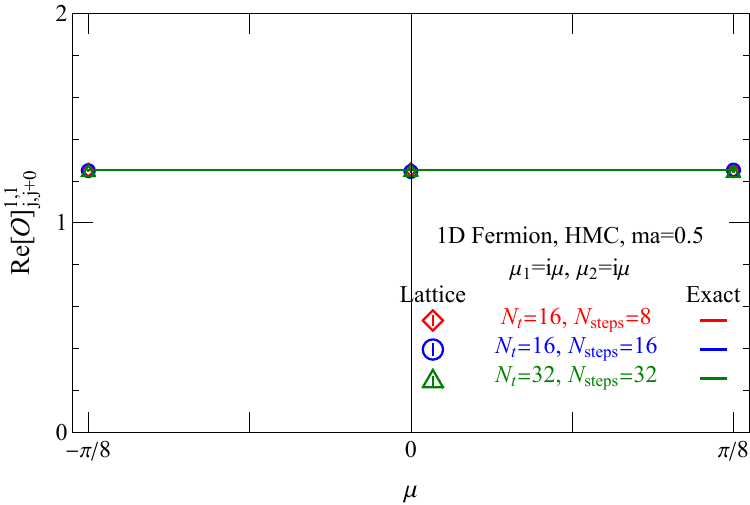}
\includegraphics[scale = 0.4]{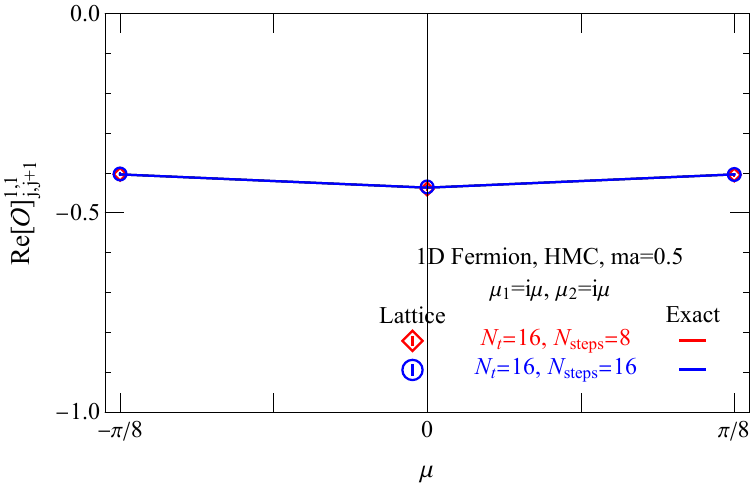}
\includegraphics[scale = 0.4]{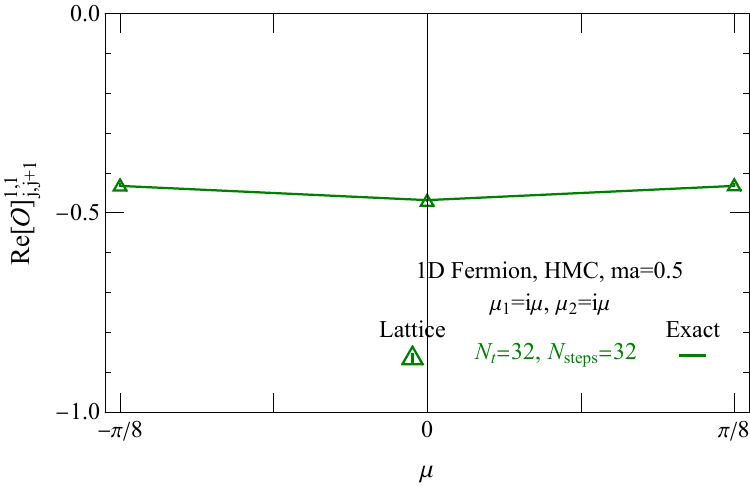}
\includegraphics[scale = 0.4]{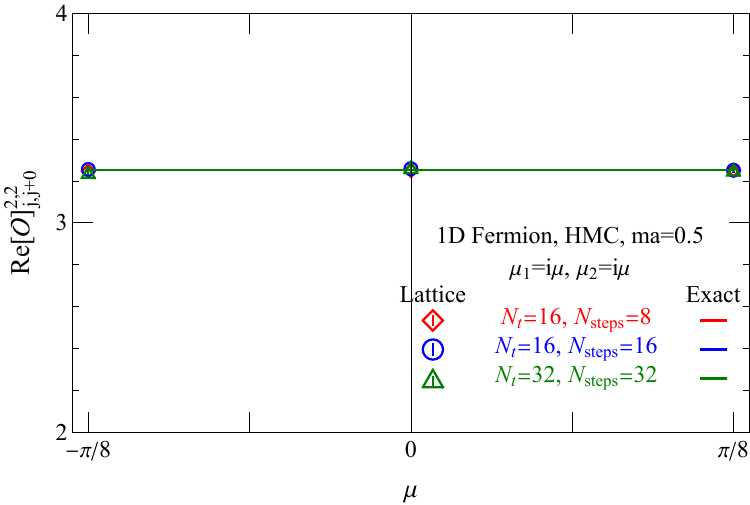}
\includegraphics[scale = 0.4]{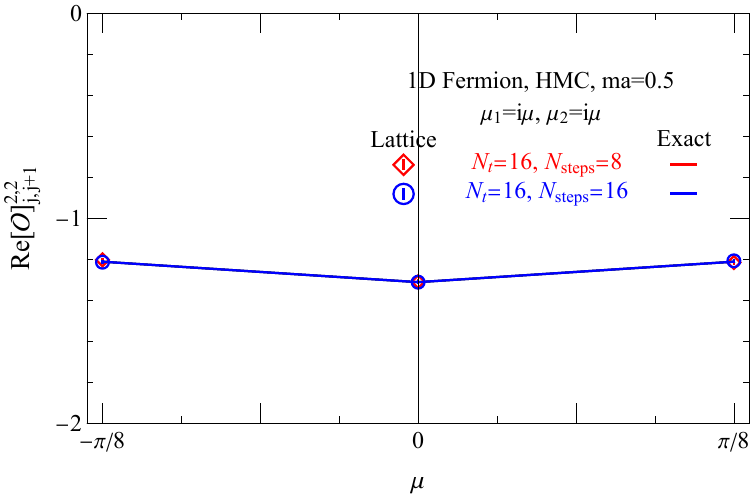}
\includegraphics[scale = 0.4]{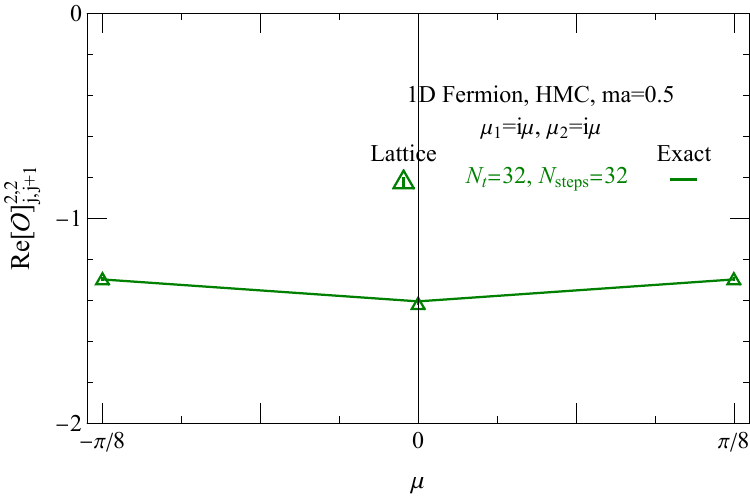}
\caption{
We use the Hybrid Monte Carlo (HMC) algorithm to compute $\mathrm{Re}(O^{\alpha_1,\alpha_2}_{j,j+n})$ for imaginary chemical potentials $(\mu_1, \mu_2)=(i\mu, i\mu)$ with $ma=0.5$, $(\alpha_1, \alpha_2)=(1, 1), (2, 2)$, and $n=0, 1$ for $N_t=16$ and $32$.
The simulation uses $2^{14}$ measurement sweeps, $2^6$ thermalization sweeps, and a measurement interval of $2^5$ sweeps.
The error bars are less than $1\%$.
Here, $N_{\mathrm{steps}}$ denotes the number of molecular-dynamics steps.
}
\label{pf2}
\end{figure}
\begin{figure}[tbp]
\centering
\includegraphics[scale = 0.6]{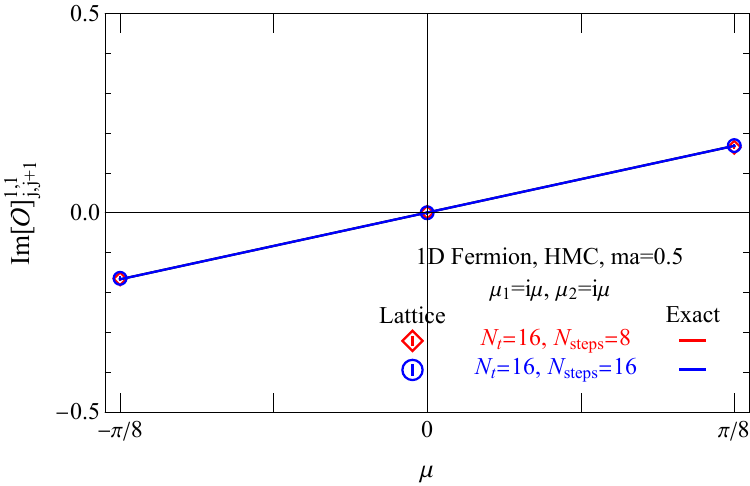}
\includegraphics[scale = 0.6]{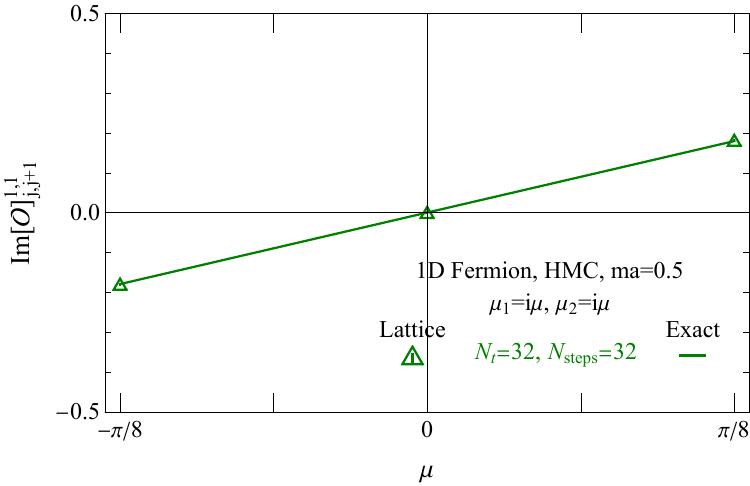}
\includegraphics[scale = 0.6]{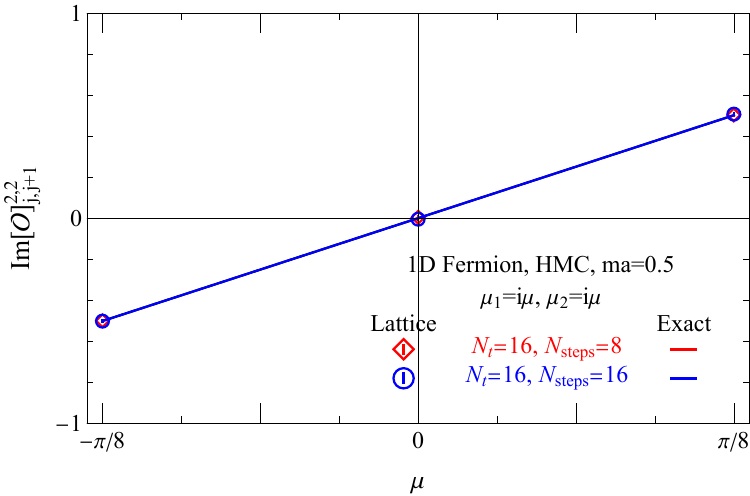}
\includegraphics[scale = 0.6]{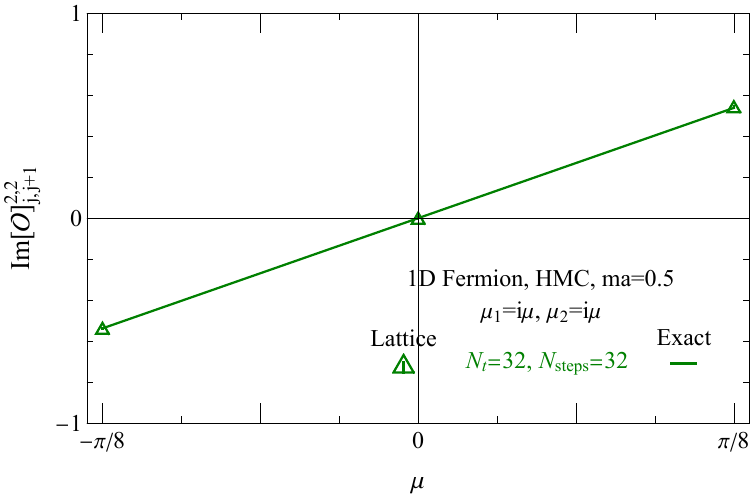}
\caption{
We use the Hybrid Monte Carlo (HMC) algorithm to compute $\mathrm{Im}(O^{\alpha_1,\alpha_2}_{j,j+1})$ for imaginary chemical potentials $(\mu_1, \mu_2)=(i\mu, i\mu)$ with $ma=0.5$ and $(\alpha_1, \alpha_2)=(1, 1), (2, 2)$ for $N_t=16$ and $32$.
The simulation uses $2^{14}$ measurement sweeps, $2^6$ thermalization sweeps, and a measurement interval of $2^5$ sweeps.
The error bars are less than $1\%$.
Here, $N_{\mathrm{steps}}$ denotes the number of molecular-dynamics steps.
}
\label{pf3}
\end{figure}

\section{AI-Assisted Analytic Continuation}
\label{sec:6}
\noindent
We use AI-assisted fitting to identify a suitable Laurent exponential model for data at imaginary chemical potential and to combine the fitted functions to real chemical potential.
We then construct two neural-network realizations: a physics-constrained neural network (PCNN) and a physics-neural hybrid model.
Because the PCNN already achieves near-machine-precision accuracy on the present data, the hybrid model yields essentially the same results and offers no further improvement.
Its usefulness should therefore be assessed on more complicated or deliberately incomplete lattice-model data.

\subsection{Laurent Exponential Model and Physics-Constrained Neural Network}
\noindent
To exploit the analytic structure of the lattice solution, we construct a physics-constrained neural network (PCNN) whose hidden features are prescribed by the analytic approximation rather than learned from data.
Unlike conventional feed-forward neural networks, where the hidden-layer representations are determined through optimization, the present architecture fixes the hidden features according to the exponential modes generated by the lattice operator.
Consequently, the network learns only the amplitudes associated with these physically motivated modes, thereby incorporating prior analytic knowledge directly into the model architecture (Fig. \ref{PCNN_architecture}).
For the present first-order approximation to the one-sided lattice derivatives, the analytic feature layer is chosen as
\bea
\phi(z)=[ e^{-z} \ 1 \ e^z],
\eea
which consists of the three dominant exponential modes predicted by the lattice analysis.
The corresponding trainable output layer is
\bea
f_{\mathrm{PCNN}}(z)=\phi(z)w=w_{-1}e^{-z}+w_0+w_1e^z, \ w\in\mathbb{R}^3,
\eea
where the trainable weights are
\bea
w=\begin{pmatrix}
w_{-1}
\\
w_0
\\
w_1
\end{pmatrix},
\eea
which is the vector of trainable parameters.
\begin{figure}[tbp]
\centering
\includegraphics[scale = 0.42]{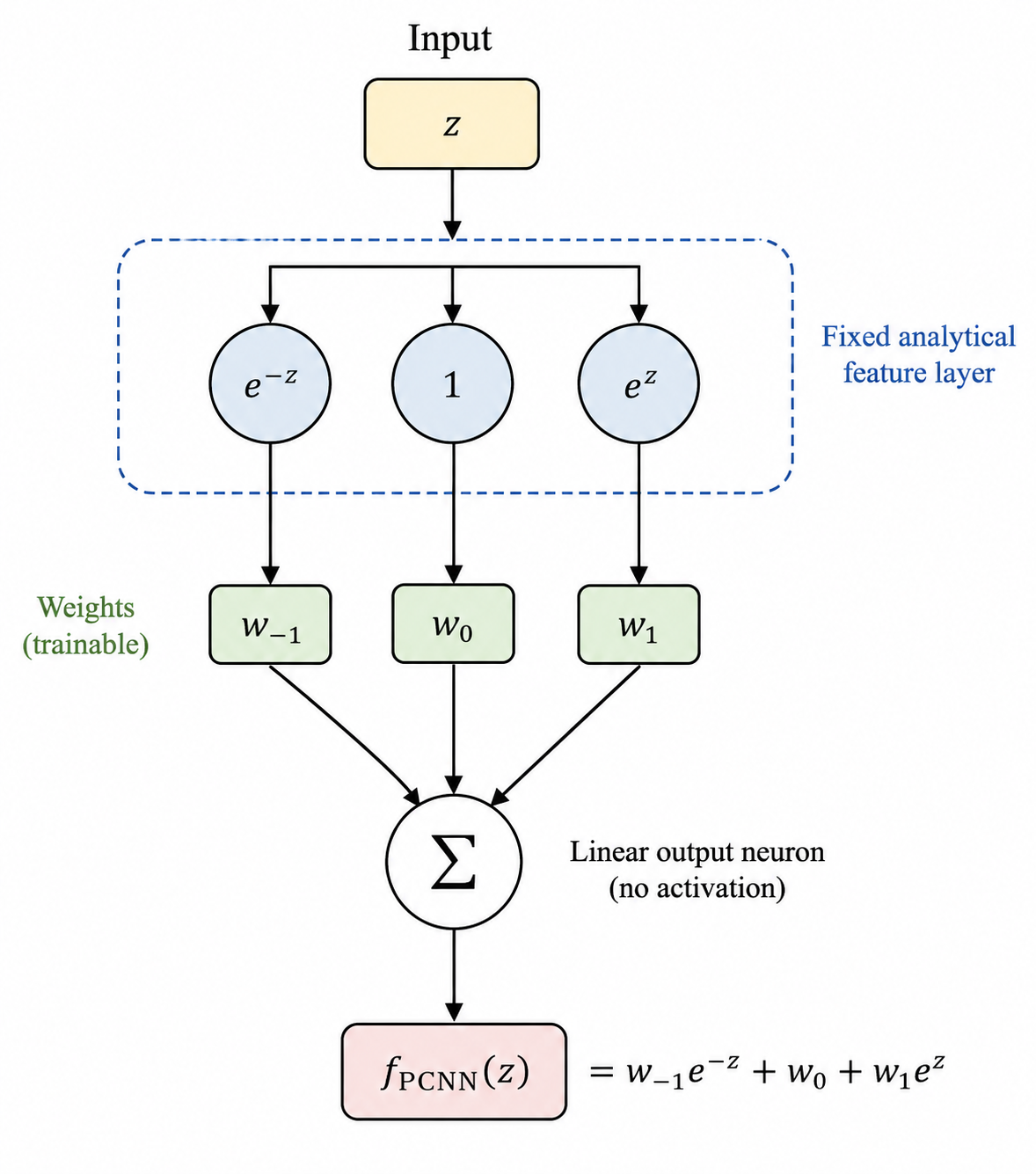}
\caption{Architecture of the proposed physics-constrained neural network (PCNN).
The input $z$ is mapped to a fixed analytic feature layer consisting of the exponential modes $\{\exp(-z), 1, \exp(z)\}$, which are linearly combined by the trainable weights $\{w_{-1}, w_0, w_1\}$ to produce the output $f_{\mathrm{PCNN}}(z)$.
}
\label{PCNN_architecture}
\end{figure}
\\

\noindent
The three feature neurons have clear physical interpretations.
The neuron $\exp(-z)$ represents the backward exponential mode, the constant neuron represents the zero-order mode, and the neuron $\exp(z)$ represents the forward exponential mode.
These three modes correspond to the dominant solutions suggested by the lattice operator.
Therefore, the hidden layer does not perform unconstrained feature extraction, but instead provides an analytically prescribed coordinate system for representing the lattice solution.
The trainable parameters determine the relative amplitudes of these physically meaningful modes.
\\

\noindent
Although the hidden features are fixed, the proposed model is still a neural network because it consists of a feature layer, trainable weights, a differentiable forward mapping, and weight optimization using standard neural-network algorithms.
The principal distinction from conventional neural networks is that the feature layer is designed analytically rather than learned from the training data.
Consequently, the present architecture can be regarded as a physics-constrained neural network with an analytically designed hidden representation.
\\

\noindent
The proposed architecture is closely related to the first-order Laurent exponential model,
\bea
f_{\mathrm{L}}(z)=a_{-1}e^{-z}+a_0+a_1 e^z.
\eea
Identifying
\bea
a_{-1}=w_{-1}, \
a_0=w_0, \
a_1=w_1,
\eea
shows that both formulations span the same hypothesis space,
\bea
{\cal H}=\mathrm{span}\{e^{-z}, 1, e^z\}.
\eea
Therefore, the PCNN does not introduce a new approximation space.
Instead, it reformulates the Laurent exponential expansion as a differentiable neural-network architecture.
When both models are trained on the same data using the same least-squares objective, they determine the same coefficients and therefore produce the same fitted function.
Their difference lies solely in the numerical implementation used to estimate these coefficients.
The Laurent exponential model typically determines the coefficients using a direct linear least-squares solver,
\bea
\min_a\sum_j\bigg(y_j-\sum_{n=-1}^1a_n e^{nz_j}\bigg)^2.
\eea
In contrast, the PCNN treats them as trainable weights optimized within a neural-network framework.
Instead of solving
\bea
\frac{\partial L}{\partial w}=0, \ L(w)=\frac{1}{2}\sum_j\big(f(z_j)-y_j\big)^2,
\eea
we use gradient descent
\bea
w^{(k+1)}=w^{(k)}-\eta\partial_w L.
\eea
If gradient descent converges to the least-squares minimum, the resulting coefficients satisfy $w_n^*=a_n$.
\\

\noindent
The proposed PCNN possesses several attractive properties.
First, the network is completely analytical because every hidden feature has a clear mathematical origin.
Second, only three trainable parameters are required for the present first-order approximation of the one-sided lattice derivatives, making the optimization problem extremely compact and highly interpretable.
Third, the restricted hypothesis space suppresses spurious high-frequency oscillations and reduces the tendency to fit noise, since the network is constrained to combinations of physically meaningful exponential modes.
These properties improve the interpretability and robustness of the resulting approximation.
\\

\noindent
The primary limitation of the proposed architecture is its dependence on prior physical knowledge.
Because the hidden feature layer is fixed, the network cannot represent functions outside the prescribed feature space.
For example, if the true solution contains additional modes such as
\bea
e^{-2z}, \ e^{2z},
\eea
or other independent basis functions, these contributions cannot be reproduced solely by adjusting the three output weights.
In this situation, the analytical feature layer must be expanded accordingly.
A natural generalization is
\bea
\phi_K(z)=[e^{-Kz} \ \cdots \ e^{-z} \ 1 \ e^z \ \cdots \ e^{Kz}],
\eea
which leads to
\bea
f_{\mathrm{PCNN}}^{(K)}(z)=\sum_{n=-K}^{K}w_ne^{nz}.
\eea
Increasing the truncation order $K$ systematically enlarges the hypothesis space while preserving the physical interpretation of every trainable parameter.
Therefore, the approximation can be imposed in a controlled manner without altering the overall neural network architecture.
\\

\noindent
For the present data, however, the first-order feature layer already spans the same function space as the $K=1$ Laurent exponential model.
Consequently, the proposed PCNN should be viewed not as a distinct approximation model but as a neural-network realization of the Laurent exponential expansion, providing a convenient and extensible framework for larger parameter spaces.
For the current three-parameter fit, direct linear least squares is simpler and more transparent; gradient-based optimization becomes more useful when the model is extended or embedded in a larger differentiable framework.

\subsection{Physics-Neural Hybrid Model}
\noindent
The principal objective of the proposed hybrid framework is not to replace the underlying physical model with a neural network, but to preserve the available physical knowledge while introducing sufficient flexibility to compensate for its incompleteness.
Many physical problems admit a natural decomposition of the target function into a theoretically understood component and an unknown correction
\bea
f(z)=f_{\mathrm{phys}}(z)+\delta f(z),
\eea
where $f_{\mathrm{phys}}(z)$ denotes the prediction obtained from physical principles and $\delta f(z)$ represents the remaining discrepancy.
Such discrepancies may arise from incomplete theoretical descriptions, truncated perturbative expansions, finite-volume corrections, discretization effects, statistical fluctuations, or experimental uncertainties.
\\

\noindent
Instead of relearning the entire target function from data, the proposed hybrid framework preserves the physically motivated backbone and employs the neural network only to approximate the residual,
\bea
\delta f(z)\approx \delta f_{\mathrm{NN}}(z).
\eea
Consequently, the dominant analytic structure remains explicitly determined by physics.
At the same time, the neural network is responsible only for the missing information.
This decomposition improves interpretability and reduces the optimization burden by preventing the neural network from reconstructing known physical structures.

\subsubsection{Residual Learning and Approximation Strategy}
\noindent
The hybrid prediction is defined by
\bea
f_{\mathrm{hyb}}(z)=f_{\mathrm{phys}}(z)+\delta f_{\mathrm{NN}}(z),
\eea
where
\bea
\delta f_{\mathrm{NN}}(z)=d+\sum_{h=1}^Ha_h\tanh\bigg(b_h\frac{z}{\pi}+c_h\bigg)
\eea
with $a_h, b_h, c_h, d\in \mathbb{R}$.
Residual learning offers several computational advantages.
Since the dominant physical behavior has already been incorporated into $f_{\mathrm{phys}}(z)$, the residual generally possesses a smaller magnitude and a smoother functional dependence than the original target function.
Consequently, fewer trainable parameters are required to achieve a comparable approximation accuracy, and the optimization landscape becomes substantially simpler than that of directly approximating the complete function.
Moreover, the magnitude of the learned residual provides a direct measure of the physical model's adequacy.
A negligible residual indicates that the available physical constraints are already sufficient to describe the data, whereas a significant residual signals the presence of missing physical effects or systematic discrepancies.
\\

\noindent
Since analytic continuation is fundamentally governed by complex analyticity, preserving the analytic structure of the approximation is more important than minimizing interpolation errors alone.
Unlike conventional complex-valued neural networks that independently approximate the real and imaginary parts,
\bea
f(z)=u(x, y)+iv(x, y),
\eea
the present framework directly constructs the residual from holomorphic basis functions.
Consequently, the Cauchy-Riemann equations are automatically satisfied throughout the fitting domain.
\\

\noindent
The complex hyperbolic tangent is chosen as the activation function because $\tanh(z)$ is meromorphic, with isolated poles located at
\bea
z=i\pi\bigg(n+\frac{1}{2}\bigg), \ n\in\mathbb{Z}.
\eea
The affine transformation $b_hz/\pi+c_h$ translates and rescales these poles while preserving their isolated nature.
Therefore, provided that the poles remain outside the fitting region, the residual network remains holomorphic throughout the domain of interest.
This property distinguishes the present framework from conventional neural networks, whose outputs generally do not satisfy the Cauchy-Riemann equations and therefore cannot be interpreted as analytic continuations.
\\

\noindent
The choice of the complex hyperbolic tangent is motivated by several considerations.
First, it preserves complex analyticity except at isolated poles.
Second, its nonlinear behavior can represent localized residual structures while avoiding the severe oscillations often associated with high-order polynomial approximations.
Third, its analytic derivative permits stable gradient-based optimization.
Finally, its meromorphic structure is mathematically well understood, making it suitable for applications involving analytic continuation.
Although other holomorphic activation functions could also be employed, the complex hyperbolic tangent provides an effective compromise between expressive power, computational efficiency, and numerical stability.

\subsubsection{Optimization and Regularization}
\noindent
The trainable parameters
\bea
\theta=\{a_h, b_h, c_h, d\}
\eea
are obtained by minimizing
\bea
{\cal L}=\frac{1}{N}\sum_{j=1}^N|f_{\mathrm{hyb}}(z_j)-y_j|^2
+\lambda\bigg(d^2+\sum_{h=1}^H(a_h^2+b_h^2+c_h^2)\bigg).
\eea
The coefficient $\lambda$ is treated as a tunable hyperparameter rather than a universal constant.
Setting
\bea
\lambda=0
\eea
corresponds to training without L2 regularization.
Increasing $\lambda$ progressively suppresses large parameter values and therefore reduces the complexity of the learned residual.
The optimal value of $\lambda$ depends on the expected noise level, the complexity of the residual function, and the desired trade-off between approximation accuracy and generalization.
In practical applications, $\lambda$ can be selected using validation data, cross-validation, Bayesian optimization, or prior physical knowledge.
\\

\noindent
In addition to reducing overfitting, L2 regularization improves the stability of analytic continuation.
Large network parameters can generate rapidly varying analytic functions or move the poles of the complex activation closer to the fitting domain.
Penalizing large coefficients therefore discourages unnecessarily complicated analytic structures and promotes smoother extrapolation.

\subsubsection{Pole Structure and Numerical Stability}
\noindent
The locations of the poles of the residual network depend on the trainable parameters through the affine transformation $b_hz/\pi+c_h$.
In the present implementation, no explicit constraints are imposed on the pole locations.
Instead, numerical stability is improved by choosing an appropriate fitting domain and applying L2 regularization, which suppresses excessively large parameter values and reduces the tendency for poles to approach the sampled region.
No numerical instability associated with nearby poles is observed in the benchmark calculations.
Future implementations should incorporate explicit pole constraints or regularization terms that directly penalize poles entering the fitting domain.
\\

\noindent
The neural correction is intended to capture systematic discrepancies rather than numerical roundoff errors.
Therefore, the hybrid framework employs an adaptive gate based on the prediction accuracy of the physics backbone.
Let $E_{\mathrm{phys}}$ denote the normalized root-mean-square error of the physics backbone.
The hybrid prediction is defined as
\bea
f_{\mathrm{hyb}}(z)=\Bigg\{\begin{array}{ll}
f_{\mathrm{phys}}(z), & E_{\mathrm{phys}}< \epsilon, \\
f_{\mathrm{phys}}(z)+\delta f_{\mathrm{NN}}(z), & E_{\mathrm{phys}}\ge \epsilon.
\end{array}.
\eea
The threshold
\bea
\epsilon=10^{-10}
\eea
used in this paper is regarded as an algorithmic tolerance rather than a universal constant.
Different applications require different tolerances depending on the numerical precision of the data and the desired prediction accuracy.
Consequently, the gate implements adaptive model selection.
Whenever the physical model already provides a sufficiently accurate description, introducing additional neural degrees of freedom would primarily fit floating-point roundoff errors rather than meaningful physical information.
The neural correction is therefore activated only when residual structures remain above the chosen numerical evidence.

\subsubsection{Theoretical Guarantees}
\noindent
The construction of the hybrid model immediately implies several mathematical properties.
\begin{proposition} [\textbf{Holomorphicity}]
Let $D\in\mathbb{C}$ be an open domain.
Suppose that $f_{\mathrm{phys}}$ is holomorphic on $D$ and that every pole of $\delta f_{\mathrm{NN}}$ lies outside $D$.
We then get that
\bea
f_{\mathrm{hyb}}(z)=f_{\mathrm{phys}}(z)+\delta f_{\mathrm{NN}}(z)
\eea
is holomorphic on $D$.
\end{proposition}
\begin{proof}
Since both $f_{\mathrm{phys}}$ and $\delta f_{\mathrm{NN}}$ are holomorphic on $D$, their sum is holomorphic by the closure of holomorphic functions under addition.
\end{proof}
\begin{proposition} [\textbf{Consistency}]
If the physical model exactly reproduces the target function,
\bea
f_{\mathrm{target}}(z)=f_{\mathrm{phys}}(z),
\eea
the optimal residual satisfies
\bea
\delta f_{\mathrm{NN}}(z)=0,
\eea
and the hybrid model reduces to the physics backbone.
\end{proposition}
\begin{proof}
The residual is identically zero,
\bea
f_{\mathrm{target}}-f_{\mathrm{phys}}=0.
\eea
Since the loss function is non-negative, its global minimum is attained by the zero residual.
\end{proof}
\noindent
These propositions clarify that the hybrid framework is consistent with the underlying physical model and preserves analyticity by construction.

\subsubsection{Computational Complexity}
\noindent
Because the neural network learns only the residual, the number of hidden neurons is typically much smaller than that required to approximate the complete target function.
Consequently, the computational overhead introduced by the hybrid framework is generally modest.
The dominant physical behavior is already captured by the PCNN, allowing the optimization to focus on comparatively small corrections.
\\

\noindent
The present framework assumes that the unknown correction is sufficiently smooth to be represented by a finite holomorphic neural network.
More complicated singularities, such as branch cuts or essential singularities not incorporated into the physics backbone, would require alternative architectures or explicit singular-function decompositions.
Furthermore, all benchmark problems considered in this paper are already reproduced by the PCNN with a normalized root-mean-square error of approximately $10^{-13}$, so the adaptive gate remains inactive.
Consequently, the present study primarily validates the design of the hybrid framework rather than demonstrating the practical improvement provided by the neural residual.
A more comprehensive evaluation should therefore consider noisy datasets, deliberately incomplete physical models, or experimental measurements, in which the residual network is expected to become active, and its contribution can be quantitatively assessed.

\subsection{From Imaginary Chemical Potential to Real Chemical Potential}
\noindent
We demonstrate AI-assisted analytic continuation of the two-point pseudo-fermion correlators ${\cal O}_{j, j+1}^{\alpha\beta}$ from imaginary chemical potential to real chemical potential.
Specifically, we continue $O^{1, 1}_{j,j+1}$ (in Figs. \ref{fit1}, \ref{fit2}) and $O^{2, 2}_{j,j+1}$ (in Figs. \ref{fit3}, \ref{fit4}).
The Laurent exponential model and the PCNN reproduce the data accurately with only a small number of training points, for both lattice sizes.
For $(\alpha, \beta)=(1, 1)$, the fitted functions on the imaginary axis are $-0.4375\times\exp(-i\mu)$ for $N_t=16$ and $-0.46875\times\exp(-i\mu)$ for $N_t=32$.
For $(\alpha, \beta)=(2, 2)$, they are $-1.3125\times\exp(-i\mu)$ for $N_t=16$ and $-1.40625\times\exp(-i\mu)$ for $N_t=32$.
After training, we perform analytic continuation by replacing $\exp(-i\mu)$ with $\exp(-\mu)$.
The continued functions agree well with the exact real-axis results.
These results demonstrate the usefulness of the AI-assisted fitting framework for this exactly solvable benchmark.
\begin{figure}[tbp]
\centering
\includegraphics[scale = 0.62]{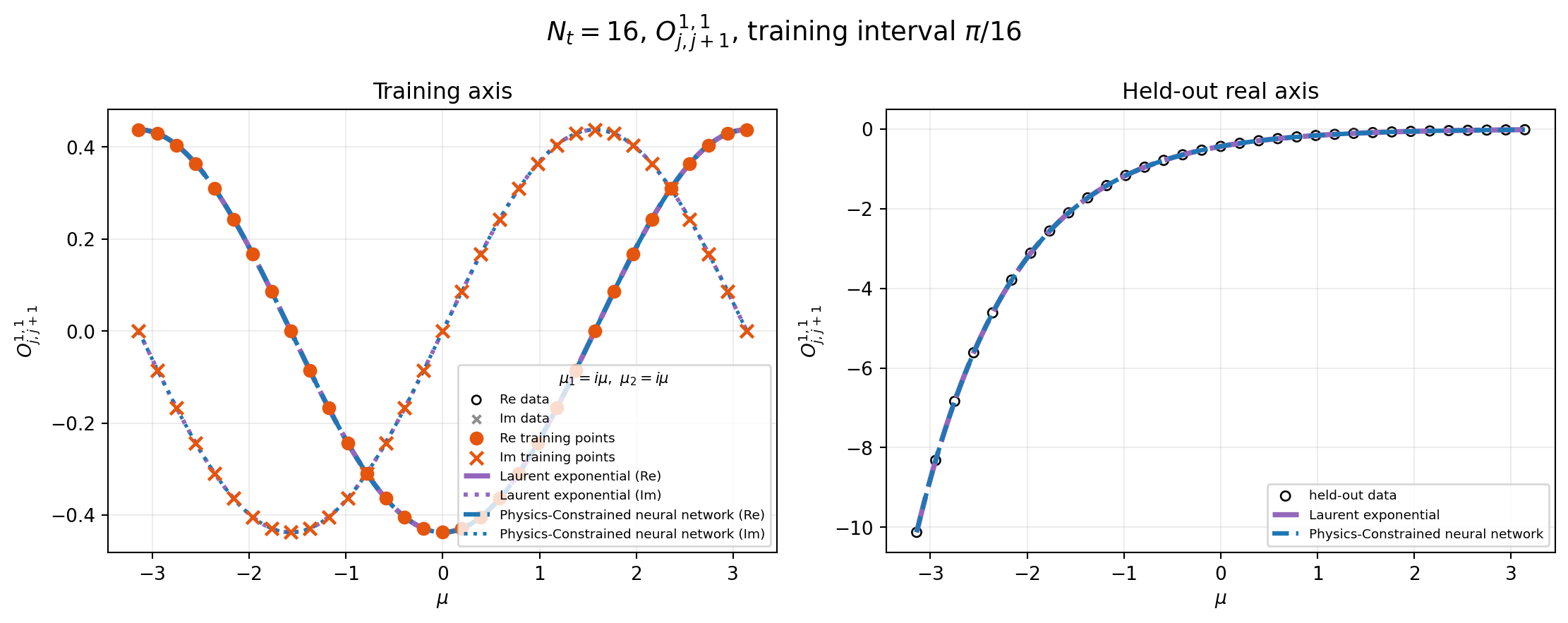}
\includegraphics[scale = 0.62]{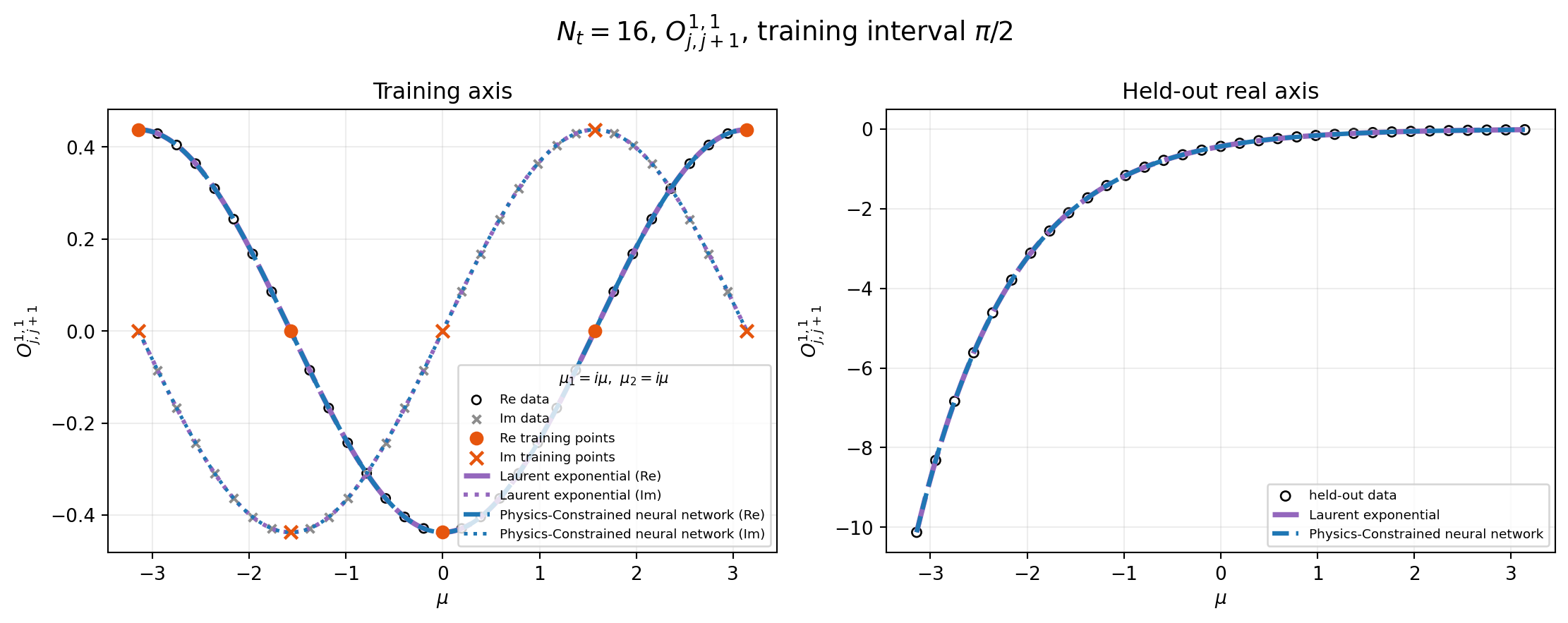}
\caption{
Analytic continuation from imaginary chemical potential (left panels) to real chemical potential (right panels) for $O^{1, 1}_{j,j+1}$ using the Laurent exponential model and the physics-constrained neural network (PCNN) at $N_t=16$.
The fits obtained with training intervals $\pi/16$ and $\pi/2$ are nearly identical.
}
\label{fit1}
\end{figure}
\begin{figure}[tbp]
\centering
\includegraphics[scale = 0.62]{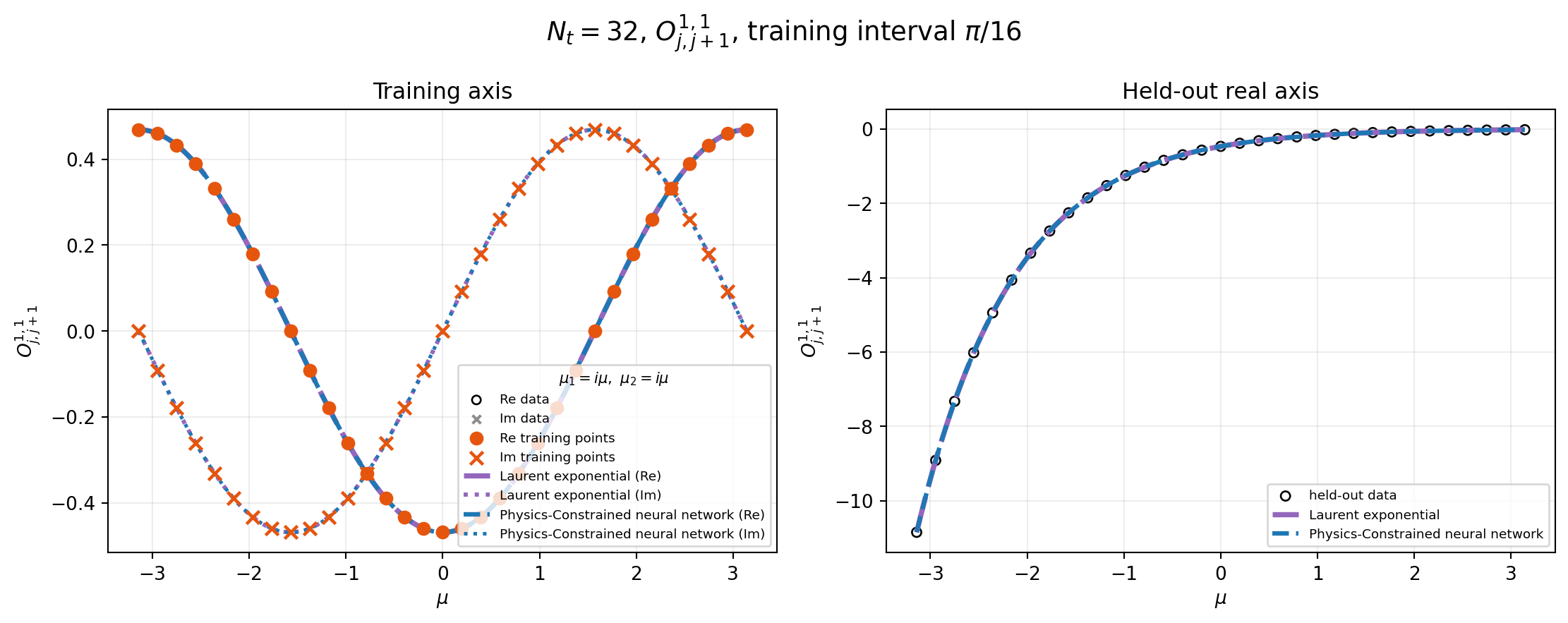}
\includegraphics[scale = 0.62]{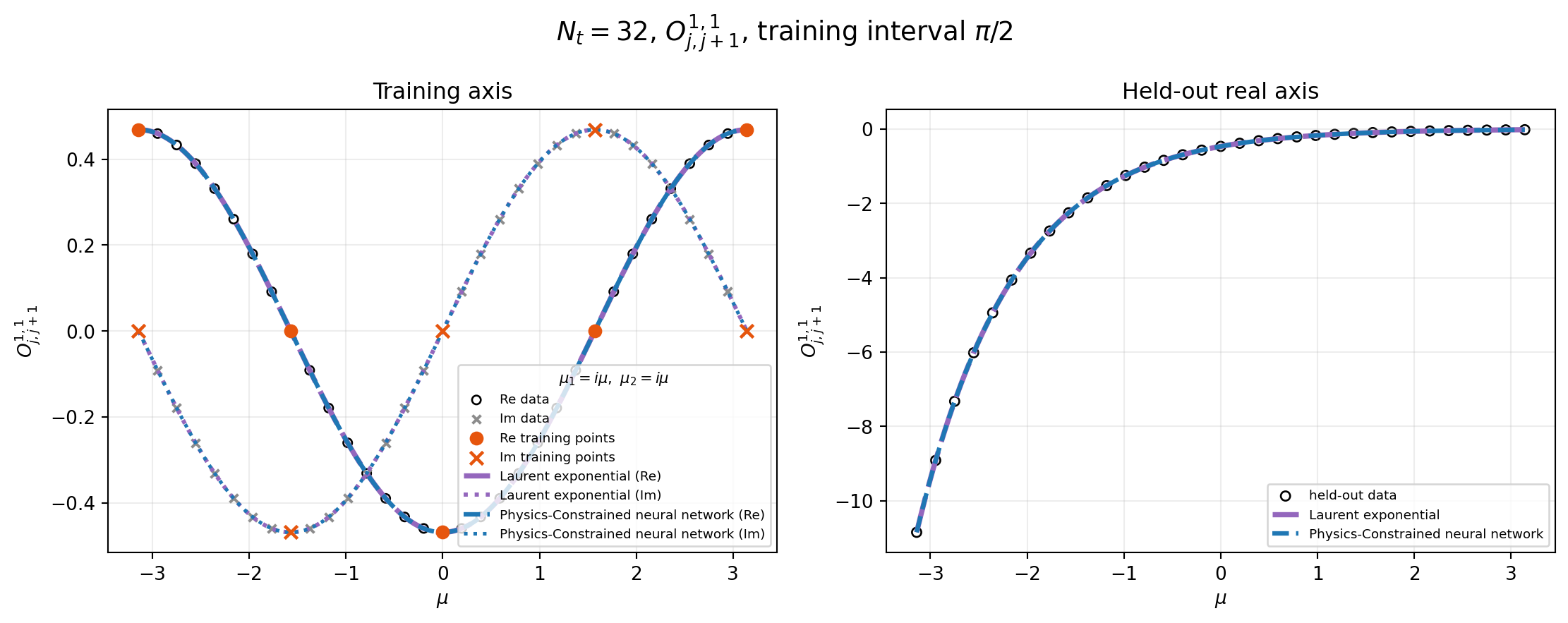}
\caption{
Analytic continuation from imaginary chemical potential (left panels) to real chemical potential (right panels) for $O^{1, 1}_{j,j+1}$ using the Laurent exponential model and the physics-constrained neural network (PCNN) at $N_t=32$.
The fits obtained with training intervals $\pi/16$ and $\pi/2$ are nearly identical.
}
\label{fit2}
\end{figure}
\begin{figure}[tbp]
\centering
\includegraphics[scale = 0.62]{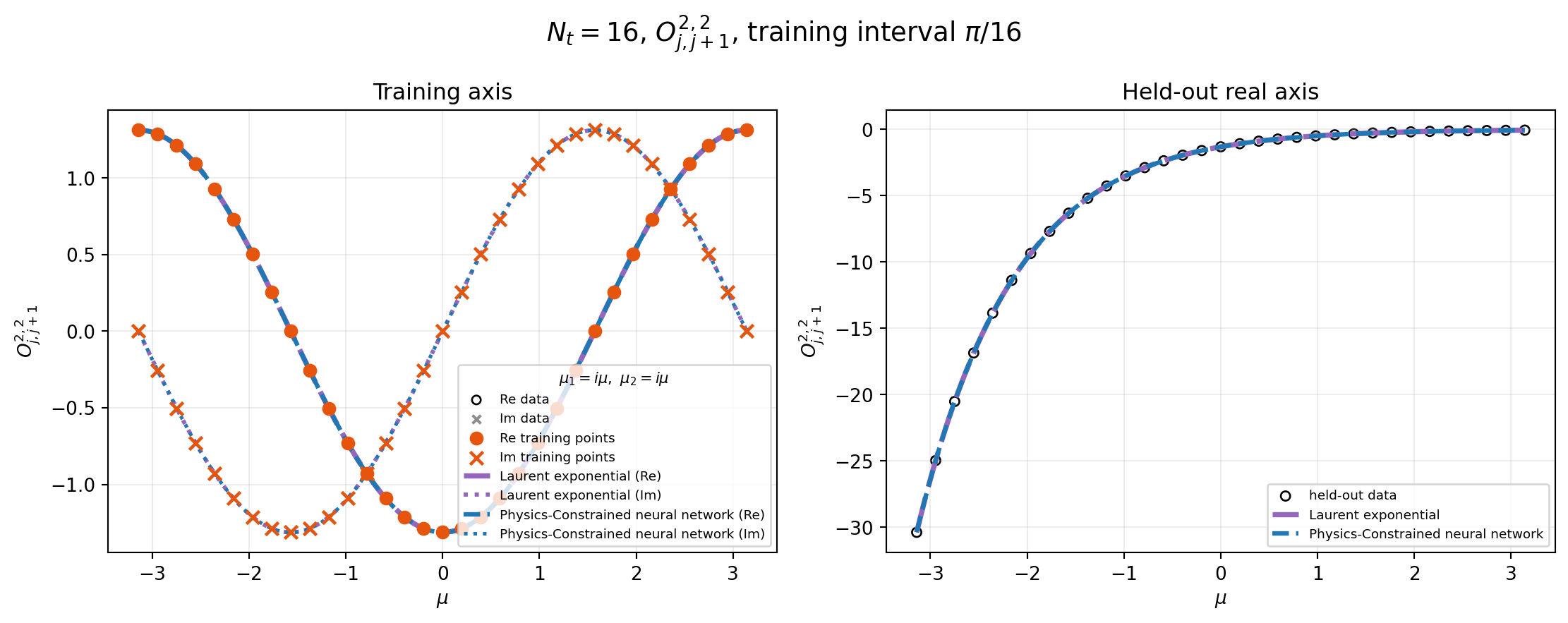}
\includegraphics[scale = 0.62]{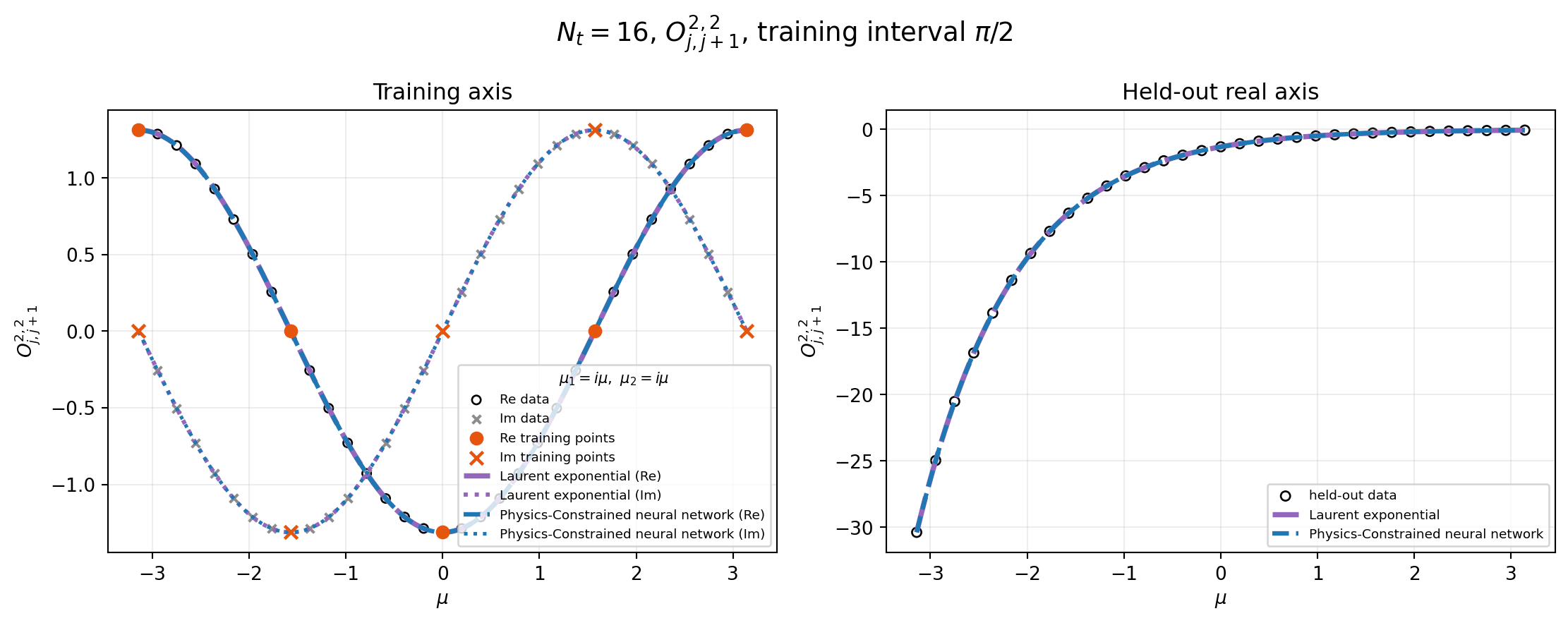}
\caption{
Analytic continuation from imaginary chemical potential (left panels) to real chemical potential (right panels) for $O^{2, 2}_{j,j+1}$ using the Laurent exponential model and the physics-constrained neural network (PCNN) at $N_t=16$.
The fits obtained with training intervals $\pi/16$ and $\pi/2$ are nearly identical.
}
\label{fit3}
\end{figure}
\begin{figure}[tbp]
\centering
\includegraphics[scale = 0.62]{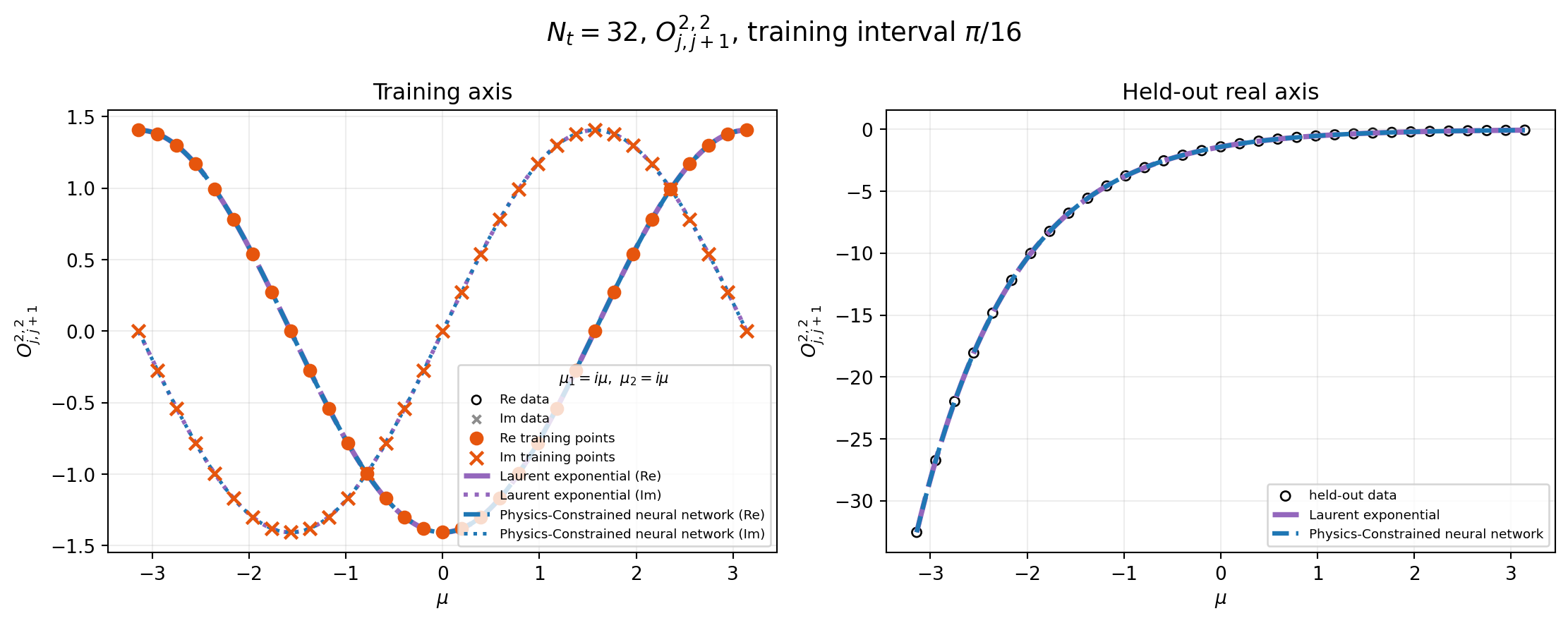}
\includegraphics[scale = 0.62]{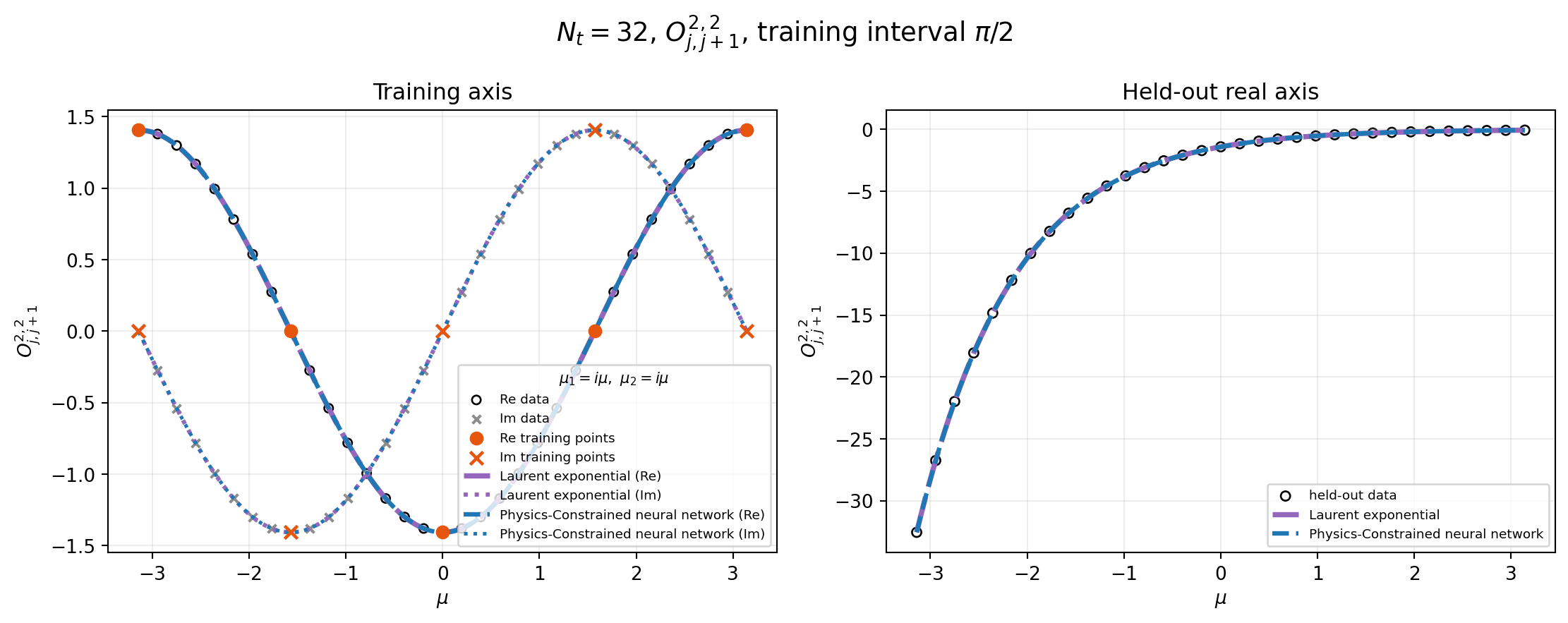}
\caption{
Analytic continuation from imaginary chemical potential (left panels) to real chemical potential (right panels) for $O^{2, 2}_{j,j+1}$ using the Laurent exponential model and the physics-constrained neural network (PCNN) at $N_t=32$.
The fits obtained with training intervals $\pi/16$ and $\pi/2$ are nearly identical.
}
\label{fit4}
\end{figure}

\newpage
\section{Discussion and Conclusion}
\label{sec:7}
\noindent
In this work, we investigated analytic continuation with respect to the chemical potential in lattice field theory.
Analytic continuation from imaginary to real chemical potential is widely used to explore finite-density systems because direct lattice simulations generally suffer from the sign problem \cite{Lombardo:1999cz}.
However, its validity relies on the analyticity of the underlying observables, which is difficult to establish in general.
To address this issue, we studied one-dimensional free Dirac fermions, for which exact solutions are available.
We considered both the naive and exponential implementations of the chemical potential.
We showed that, on a finite lattice, analytic continuation is valid.
The infinite-volume limit can destroy analyticity, implying that the limit and analytic continuation do not generally commute. In addition, the exponential implementation preserves the correct continuum limit without imposing further constraints on the chemical potential, providing theoretical support for its use in lattice simulations.
\\

\noindent
Although the lattice formulation is non-Hermitian \cite{Stamatescu:1993ga}, we showed that degenerate fermions with paired forward and backward lattice derivatives admit a non-negative formulation for the chemical-potential assignments $(\mu_1,\mu_2)=(\mu,-\mu)$ and $(i\mu, i\mu)$, enabling Hybrid Monte Carlo simulations.
The numerical results for the two-point pseudofermion correlation functions agree well with the exact solutions for $N_t=16$ and $32$.
Furthermore, by analyzing the eigenvalue spectrum of the Dirac operator, we showed that, in the cases studied here, the sign problem is a technical obstacle to introducing pseudofermions rather than an indication of an ill-defined fermionic theory.
For two degenerate flavors, the occurrence of a sign problem coincides with the appearance of Dirac eigenvalues with negative real parts in the examined configurations.
We further demonstrated that AI-assisted analytic continuation based on the Laurent exponential model and a physics-constrained neural network accurately reconstructs the two-point correlation functions at real chemical potential from data obtained at imaginary chemical potential.
Although the optimal fitting ansatz depends on the observable, the proposed framework provides a practical strategy for this benchmark and a potentially useful starting point for more general finite-density lattice systems.
\\

\noindent
The exact solution of the one-dimensional model reveals that the relevant finite-lattice observables remain analytic over a broader domain than their infinite-volume limit.
Consequently, analytic continuation over the finite-lattice domain should be performed before taking the limit, with the order of limits of limits stated explicitly.
This observation provides a theoretical foundation for analytic-continuation techniques widely used in finite-density lattice field theory and clarifies their range of validity.
\\

\noindent
An important future direction is to extend the present framework beyond free fermions.
In particular, it will be interesting to investigate interacting theories and determine whether the eigenvalue criterion identified here continues to characterize the sign problem.
Since interactions generally produce complex Dirac spectra, understanding the relation between the eigenvalue and the sign problem may provide a useful classification of finite-density lattice theories.
Ultimately, our goal is to apply these ideas to finite-density QCD \cite{Lombardo:1999cz}.
As an intermediate step, we plan to study the Gross--Neveu--Yukawa model and then generalize the present HMC algorithm to non-degenerate fermion masses, for example by introducing the twisted-mass formulation and first testing it in one-dimensional free Dirac fermions.

\section*{Acknowledgments}
\noindent 
We thank Sinya Aoki for his helpful discussion. 
%HZ acknowledges the Guangdong Major Project of Basic and Applied Basic Research
%(Grant No. 2020B0301030008) and the National Natural Science Foundation of China
%(Grant No. 12105107). 
CTM thanks Nan-Peng Ma for his encouragement. 

%\noindent
%The author acknowledges the YST Program of the APCTP; 
%Post-Doctoral International Exchange Program (Grant No. YJ20180087); 
%China Postdoctoral Science Foundation, Postdoctoral General Funding: Second Class (Grant No. 2019M652926); 
%Foreign Young Talents Program (Grant No. QN20200230017); 
%Science and Technology Program of Guangzhou (Grant No. 2019050001).

%\appendix
%\section{Thermalization}
%\label{sec:A}

%\section{Auto-Correlation Time}
%\label{sec:B}


\baselineskip 22pt
\begin{thebibliography}{99}
%\cite{Feynman:1948ur}
\bibitem{Feynman:1948ur}
R.~P.~Feynman,
``Space-time approach to nonrelativistic quantum mechanics,''
Rev. Mod. Phys. \textbf{20}, 367-387 (1948)
doi:10.1103/RevModPhys.20.367
%1898 citations counted in INSPIRE as of 27 Jul 2026

%\cite{Wilson:1974sk}
\bibitem{Wilson:1974sk}
K.~G.~Wilson,
``Confinement of Quarks,''
Phys. Rev. D \textbf{10}, 2445-2459 (1974)
doi:10.1103/PhysRevD.10.2445
%7152 citations counted in INSPIRE as of 27 Jul 2026

%\cite{Makeenko:2002uj}
\bibitem{Makeenko:2002uj}
Y.~Makeenko,
``Methods of contemporary gauge theory,''
Cambridge University Press, 2005,
ISBN 978-0-521-02215-6, 978-0-521-80911-5, 978-0-511-05768-7
doi:10.1017/CBO9780511535147
%25 citations counted in INSPIRE as of 27 Jul 2026

%\cite{Gattringer:2010zz}
\bibitem{Gattringer:2010zz}
C.~Gattringer and C.~B.~Lang,
``Quantum chromodynamics on the lattice,''
Lect. Notes Phys. \textbf{788}, 1-343 (2010)
Springer, 2010,
ISBN 978-3-642-01849-7, 978-3-642-01850-3
doi:10.1007/978-3-642-01850-3
%476 citations counted in INSPIRE as of 27 Jul 2026

%\cite{Ma:2024zbl}
\bibitem{Ma:2024zbl}
C.~T.~Ma and H.~Zhang,
``Lattice chiral fermion without Hermiticity,''
Int. J. Mod. Phys. A \textbf{40}, no.23, 2530008 (2025)
doi:10.1142/S0217751X2530008X
[arXiv:2411.09886 [hep-lat]].
%2 citations counted in INSPIRE as of 27 Jul 2026

%\cite{Nielsen:1980rz}
\bibitem{Nielsen:1980rz}
H.~B.~Nielsen and M.~Ninomiya,
``Absence of Neutrinos on a Lattice. 1. Proof by Homotopy Theory,''
Nucl. Phys. B \textbf{185}, 20 (1981)
[erratum: Nucl. Phys. B \textbf{195}, 541 (1982)]
doi:10.1016/0550-3213(82)90011-6
%1665 citations counted in INSPIRE as of 27 Jul 2026

%\cite{Nielsen:1981xu}
\bibitem{Nielsen:1981xu}
H.~B.~Nielsen and M.~Ninomiya,
``Absence of Neutrinos on a Lattice. 2. Intuitive Topological Proof,''
Nucl. Phys. B \textbf{193}, 173-194 (1981)
doi:10.1016/0550-3213(81)90524-1
%1112 citations counted in INSPIRE as of 27 Jul 2026

%\cite{Karsten:1981gd}
\bibitem{Karsten:1981gd}
L.~H.~Karsten,
``Lattice Fermions in Euclidean Space-time,''
Phys. Lett. B \textbf{104}, 315-319 (1981)
doi:10.1016/0370-2693(81)90133-7
%205 citations counted in INSPIRE as of 27 Jul 2026

%\cite{Wick:1954eu}
\bibitem{Wick:1954eu}
G.~C.~Wick,
``Properties of Bethe-Salpeter Wave Functions,''
Phys. Rev. \textbf{96}, 1124-1134 (1954)
doi:10.1103/PhysRev.96.1124
%675 citations counted in INSPIRE as of 27 Jul 2026

%\cite{Kosyakov:2026nrs}
\bibitem{Kosyakov:2026nrs}
B.~P.~Kosyakov, E.~Y.~Popov and M.~A.~Vronski{\u{i}},
``Can Euclidean lattice quantum field theory be analytically continued into Minkowski space?,''
Mod. Phys. Lett. A \textbf{41}, no.26, 2650146 (2026)
doi:10.1142/s0217732326501464
[arXiv:2605.18787 [hep-th]].
%0 citations counted in INSPIRE as of 27 Jul 2026

%\cite{Kogut:1974ag}
\bibitem{Kogut:1974ag}
J.~B.~Kogut and L.~Susskind,
``Hamiltonian Formulation of Wilson's Lattice Gauge Theories,''
Phys. Rev. D \textbf{11}, 395-408 (1975)
doi:10.1103/PhysRevD.11.395
%2782 citations counted in INSPIRE as of 27 Jul 2026

%\cite{Ginsparg:1981bj}
\bibitem{Ginsparg:1981bj}
P.~H.~Ginsparg and K.~G.~Wilson,
``A Remnant of Chiral Symmetry on the Lattice,''
Phys. Rev. D \textbf{25}, 2649 (1982)
doi:10.1103/PhysRevD.25.2649
%1362 citations counted in INSPIRE as of 27 Jul 2026

%\cite{Luscher:1998pqa}
\bibitem{Luscher:1998pqa}
M.~Luscher,
``Exact chiral symmetry on the lattice and the Ginsparg-Wilson relation,''
Phys. Lett. B \textbf{428}, 342-345 (1998)
doi:10.1016/S0370-2693(98)00423-7
[arXiv:hep-lat/9802011 [hep-lat]].
%977 citations counted in INSPIRE as of 27 Jul 2026

%\cite{Neuberger:1998wv}
\bibitem{Neuberger:1998wv}
H.~Neuberger,
``More about exactly massless quarks on the lattice,''
Phys. Lett. B \textbf{427}, 353-355 (1998)
doi:10.1016/S0370-2693(98)00355-4
[arXiv:hep-lat/9801031 [hep-lat]].
%912 citations counted in INSPIRE as of 27 Jul 2026

%\cite{Stamatescu:1993ga}
\bibitem{Stamatescu:1993ga}
I.~O.~Stamatescu and T.~T.~Wu,
``A New formulation of lattice gauge theory with fermions,''
CERN-TH-6631-92.
%2 citations counted in INSPIRE as of 27 Jul 2026

%\cite{Stamatescu:1994yj}
\bibitem{Stamatescu:1994yj}
I.~O.~Stamatescu and T.~T.~Wu,
``Lattice fermion formulation with one-sided derivatives,''
Nucl. Phys. B Proc. Suppl. \textbf{42}, 838-840 (1995)
doi:10.1016/0920-5632(95)00397-R
%7 citations counted in INSPIRE as of 27 Jul 2026

%\cite{Sadooghi:1996ip}
\bibitem{Sadooghi:1996ip}
N.~Sadooghi and H.~J.~Rothe,
``Continuum behavior of lattice QED, discretized with one sided lattice differences, in one loop order,''
Phys. Rev. D \textbf{55}, 6749-6759 (1997)
doi:10.1103/PhysRevD.55.6749
[arXiv:hep-lat/9610001 [hep-lat]].
%17 citations counted in INSPIRE as of 27 Jul 2026

%\cite{Guo:2021sjp}
\bibitem{Guo:2021sjp}
X.~Guo, C.~T.~Ma and H.~Zhang,
``Naive lattice fermion without doublers,''
Phys. Rev. D \textbf{104}, no.9, 094505 (2021)
doi:10.1103/PhysRevD.104.094505
[arXiv:2105.10977 [hep-lat]].
%4 citations counted in INSPIRE as of 27 Jul 2026

%\cite{Guo:2024jqt}
\bibitem{Guo:2024jqt}
X.~Guo, C.~T.~Ma and H.~Zhang,
``Non-Hermitian lattice fermions in the 2D Gross-Neveu-Yukawa model,''
Phys. Rev. D \textbf{110}, no.3, 034502 (2024)
doi:10.1103/PhysRevD.110.034502
[arXiv:2404.18441 [hep-th]].
%3 citations counted in INSPIRE as of 27 Jul 2026

%\cite{Schwinger:1958qau}
\bibitem{Schwinger:1958qau}
J.~Schwinger,
``ON THE EUCLIDEAN STRUCTURE OF RELATIVISTIC FIELD THEORY,''
Proc. Nat. Acad. Sci. \textbf{44}, no.9, 956-965 (1958)
doi:10.1073/pnas.44.9.956
%101 citations counted in INSPIRE as of 27 Jul 2026

%\cite{Osterwalder:1974tc}
\bibitem{Osterwalder:1974tc}
K.~Osterwalder and R.~Schrader,
``Axioms for Euclidean Green's Functions. 2.,''
Commun. Math. Phys. \textbf{42}, 281 (1975)
doi:10.1007/BF01608978
%719 citations counted in INSPIRE as of 27 Jul 2026

%\cite{Schwinger:1951ex}
\bibitem{Schwinger:1951ex}
J.~S.~Schwinger,
``On the Green's functions of quantized fields. 1.,''
Proc. Nat. Acad. Sci. \textbf{37}, 452-455 (1951)
doi:10.1073/pnas.37.7.452
%851 citations counted in INSPIRE as of 27 Jul 2026

%\cite{Lombardo:1999cz}
\bibitem{Lombardo:1999cz}
M.~P.~Lombardo,
``Finite density (might well be easier) at finite temperature,''
Nucl. Phys. B Proc. Suppl. \textbf{83}, 375-377 (2000)
doi:10.1016/S0920-5632(00)91678-5
[arXiv:hep-lat/9908006 [hep-lat]].
%85 citations counted in INSPIRE as of 28 Jul 2026

%\cite{DElia:2002tig}
\bibitem{DElia:2002tig}
M.~D'Elia and M.~P.~Lombardo,
``Finite density QCD via imaginary chemical potential,''
Phys. Rev. D \textbf{67}, 014505 (2003)
doi:10.1103/PhysRevD.67.014505
[arXiv:hep-lat/0209146 [hep-lat]].
%733 citations counted in INSPIRE as of 28 Jul 2026

%\cite{Braun:2012ww}
\bibitem{Braun:2012ww}
J.~Braun, J.~W.~Chen, J.~Deng, J.~E.~Drut, B.~Friman, C.~T.~Ma and Y.~D.~Tsai,
``Imaginary polarization as a way to surmount the sign problem in $Ab$ $Initio$ calculations of spin-imbalanced Fermi gases,''
Phys. Rev. Lett. \textbf{110}, 130404 (2013)
doi:10.1103/PhysRevLett.110.130404
[arXiv:1209.3319 [cond-mat.stat-mech]].
%29 citations counted in INSPIRE as of 27 Jul 2026

%\cite{Brandt:2017oyy}
\bibitem{Brandt:2017oyy}
B.~B.~Brandt, G.~Endrodi and S.~Schmalzbauer,
``QCD phase diagram for nonzero isospin-asymmetry,''
Phys. Rev. D \textbf{97}, no.5, 054514 (2018)
doi:10.1103/PhysRevD.97.054514
[arXiv:1712.08190 [hep-lat]].
%206 citations counted in INSPIRE as of 27 Jul 2026

%\cite{Hasenfratz:1983ba}
\bibitem{Hasenfratz:1983ba}
P.~Hasenfratz and F.~Karsch,
``Chemical Potential on the Lattice,''
Phys. Lett. B \textbf{125}, 308-310 (1983)
doi:10.1016/0370-2693(83)91290-X
%440 citations counted in INSPIRE as of 27 Apr 2026

%\cite{Giveon:1994fu}
\bibitem{Giveon:1994fu}
A.~Giveon, M.~Porrati and E.~Rabinovici,
``Target space duality in string theory,''
Phys. Rept. \textbf{244}, 77-202 (1994)
doi:10.1016/0370-1573(94)90070-1
[arXiv:hep-th/9401139 [hep-th]].
%1163 citations counted in INSPIRE as of 27 Jul 2026

%\cite{Ma:2018efs}
\bibitem{Ma:2018efs}
C.~T.~Ma,
``Parity Anomaly and Duality Web,''
Fortsch. Phys. \textbf{66}, no.8-9, 1800045 (2018)
doi:10.1002/prop.201800045
[arXiv:1802.08959 [hep-th]].
%16 citations counted in INSPIRE as of 27 Jul 2026

%\cite{Ma:2023krt}
\bibitem{Ma:2023krt}
C.~T.~Ma,
``AdS$_{3}$ Einstein gravity and boundary description: pedagogical review,''
Class. Quant. Grav. \textbf{41}, no.2, 023001 (2024)
doi:10.1088/1361-6382/ad17f0
[arXiv:2310.04665 [hep-th]].
%2 citations counted in INSPIRE as of 27 Jul 2026

%\cite{Ma:2025zaz}
\bibitem{Ma:2025zaz}
C.~T.~Ma,
``Non-Commutative Geometry for D-Branes in Large R-R Field Background,''
[arXiv:2510.24998 [hep-th]].
%1 citations counted in INSPIRE as of 27 Jul 2026



\end{thebibliography}
\end{document}